\documentclass[11pt]{article}
\usepackage[utf8]{inputenc}

\usepackage{sn-preamble}
\usepackage{framed}
\usepackage[most]{tcolorbox}
\colorlet{shadecolor}{orange!15}

\author{
Xi Chen\thanks{Columbia University. Email: \url{xc2198@columbia.edu}.} \and 
William Pires \thanks{Columbia University. Email: \url{wp2294@columbia.edu}.} \and 
Toniann Pitassi\thanks{Columbia University. Email: \url{tonipitassi@gmail.com}.} \and 
Rocco A. Servedio\thanks{Columbia University. Email: \url{rocco@cs.columbia.edu}.} 
}

\usepackage{tikz}
\usepackage{verbatim}
\usepackage{cancel}
\usepackage{caption}

\usetikzlibrary{decorations.pathreplacing,angles,quotes}
\usetikzlibrary{shapes.geometric}
\usetikzlibrary{patterns}

\usepackage{mathtools} %
\usepackage{bm} %

\theoremstyle{definition}

\newtheorem{assumption}{Assumption}

\begin{document}
\setcounter{section}{0}

\title{Relative-error testing of conjunctions and decision lists}%

\newcommand{\red}[1]{{\color{red} {#1}}}
\newcommand{\blue}[1]{{\color{blue} {#1}}}
\newcommand{\gray}[1]{{\color{gray} {#1}}}

\newcommand{\xnote}[1]{\footnote{{\bf \color{purple}Xi}: {#1}}}
\newcommand{\rnote}[1]{\footnote{{\bf \color{red}Rocco}: {#1}}}
\newcommand{\toni}[1]{\footnote{{\bf\color{blue} toni}: {#1}}} 
\newcommand{\wnote}[1]{\footnote{{\bf \color{teal}Will}: {#1}}}

\newcommand{\bxO}{\mathbf{x^1}}
\newcommand{\bxT}{\mathbf{x^2}}
\newcommand{\bxTh}{\mathbf{x^3}}

\newcommand{\byO}{\mathbf{y^1}}
\newcommand{\byT}{\mathbf{y^2}}
\newcommand{\byTh}{\mathbf{y^3}}

\newcommand{\bGamma}{\mathbf{\Gamma}}

\newcommand{\uth}{\bigskip \bigskip{\huge{\bf UP TO HERE} \bigskip \bigskip}}
\newcommand{\SAMP}{\mathrm{Samp}}
\newcommand{\MQ}{\mathrm{MQ}}

\newcommand{\reldist}{\mathrm{rel}\text{-}\mathrm{dist}}
\newcommand{\default}{\ensuremath{\mathsf{default}}}

\newcommand\Algphase[1]{%
\Statex\hspace*{\dimexpr-\algorithmicindent-2pt\relax}%
\Statex\hspace*{-\algorithmicindent}\textbf{#1}%
\Statex\hspace*{\dimexpr-\algorithmicindent-2pt\relax}%
}

\thispagestyle{empty}

\pagenumbering{gobble}

\maketitle

\begin{abstract}
We study the \emph{relative-error} property testing model for Boolean functions that was recently introduced in the work of \cite{CDHLNSY2024}.
In relative-error testing, the testing algorithm gets uniform random satisfying assignments as well as black-box queries to $f$, and it must  accept $f$ with high probability whenever $f$ has the property that is being tested and reject any $f$ that is \emph{relative-error} far from having the property.  Here the relative-error distance from $f$ to a function $g$ is measured with respect to $|f^{-1}(1)|$ rather than with respect to the entire domain size $2^n$ as in the Hamming distance measure that is used in the standard model; thus, unlike the standard model, relative-error testing allows us to study the testability of \emph{sparse} Boolean functions that have few satisfying assignments.
It was shown in \cite{CDHLNSY2024} that relative-error testing is at least as difficult as standard-model property testing, but for many natural and important Boolean function classes the precise relationship between the two notions is unknown.

In this paper we consider the well-studied and fundamental properties of being a \emph{conjunction} and being a \emph{decision list}.  In the relative-error setting, we give an efficient one-sided error tester for conjunctions with running time and query complexity $O(1/\epsilon)$. 

Secondly, we give a two-sided relative-error $\tilde{O}(1/\epsilon)$ tester for decision lists, matching the query complexity of the state-of-the-art algorithm in the standard model \cite{Bshouty20,DLM+:07}.

\end{abstract}

\newpage
\tableofcontents
\newpage

\setcounter{page}{1}
\pagenumbering{arabic}

\section{Introduction}

\noindent 
{\bf Background.} Over the past few decades, property testing of Boolean functions has grown into a rich and active research area (see e.g.~the extensive treatment in \cite{Goldreich17book,BY22}). Emerging from the study of program testing and self-correction \cite{BLR93,GLRSW91}, the field has developed strong connections to many other topics in theoretical computer science including complexity theory,~Boolean function analysis, and computational learning theory, see e.g. \cite{DBLP:series/lncs/12050,odonnell-book,Ron:08testlearn}.

In the ``standard'' model of Boolean function property testing, the algorithm for testing an unknown Boolean function $f: \zo^n \to \zo$ has black-box oracle (also known as ``membership query'' and will be referred to as $\MQ(f)$ in this paper) access as its only mode of access to $f$.  The goal of a tester in this standard~model~is~to distinguish between the two alternatives that (a) $f$ has some property of interest, versus (b)~$f$~is ``far'' (in the sense of disagreeing on at least an $\eps$-fraction of all the $2^n$ inputs) from every function with the property. A wide range of properties (or equivalently, classes of Boolean functions) have been studied from this perspective, including but not limited to monotone Boolean functions \cite{GGLRS,FLNRRS,CS13a, CS13b,CST14, CDST15,BB16,CWX17stoc,KMS18}, unate Boolean functions \cite{KhotShinkar16,CS16,LW19,CW19,CWX17stoc,CWX17focs,CW19}, literals \cite{PRS02}, conjunctions \cite{PRS02,GoldreichRon20}, decision lists \cite{DLM+:07,Bshouty20}, linear threshold functions \cite{MORS10}, $s$-term monotone DNF \cite{PRS02,DLM+:07,CGM11,Bshouty20}, $s$-term DNF, size-$s$ decision trees, size-$s$ branching programs, size-$s$ Boolean formulas, size-$s$ Boolean circuits, functions with Fourier degree at most $d$, $s$-sparse polynomials over $\F_2$ \cite{DLM+:07,CGM11,Bshouty20}, and many others (see e.g. Figure~1 of \cite{Bshouty20}).

As the field of Boolean function property testing has grown more mature, an interest has developed in considering alternative variants of the ``standard'' property testing model discussed above.  One early variant was the \emph{distribution-free} property testing model \cite{GGR98,HalevyKushilevitz:03}, inspired by Valiant's distribution-free PAC learning model \cite{Valiant:84}. In this model the distance between two functions $f$ and $g$ is measured according to $\Pr_{\bx \sim {\cal D}}[f(\bx) \ne g(\bx)]$, where ${\cal D}$ is an unknown and arbitrary distribution over the domain $\zo^n$.  In addition to making black-box oracle queries, distribution-free testing algorithms can also draw i.i.d.~samples from the unknown distribution~${\cal D}$. The motivation behind this model was to align property testing more closely with the influential distribution-free PAC learning model and to capture scenarios of interest in which there is an underlying non-uniform distribution over the domain.  Unfortunately, many of the results that have been established for the distribution-free PAC model are negative in nature \cite{HalevyKushilevitz:05,GlasnerServedio:09toc,CX16,CP22,CFP24}, often showing that almost as many oracle calls are needed for distribution-free testing a class ${\cal C}$ as would be required for PAC learning ${\cal C}$.
(As shown in \cite{GGR98}, the PAC learning complexity of any class ${\cal C}$ gives an easy upper bound on its distribution-free testing complexity.)

\medskip

\noindent{\bf Relative-error testing.}
Very recent work \cite{CDHLNSY2024} introduced a new model of Boolean function property testing, called \emph{relative-error testing}.  The motivation for the model comes from the observation that the standard testing framework is not well suited for testing \emph{sparse} Boolean functions, i.e.,~functions $f: \zo^n \to \zo$ that have $|f^{-1}(1)| = p 2^n$ where $p$ is very small\footnote{Such functions can be of significant interest for practical as well as theoretical reasons; since Boolean functions correspond to classification rules, a sparse Boolean function corresponds to a classification rule which outputs 1 only on positive examples of some rare but potentially desirable phenomenon.} --- perhaps as small as $1/n^{\omega(1)}$ or even $2^{-\Theta(n)}$ ---
 simply because any such function will be $p$-close to the constant-$0$ function.
 (This is closely analogous to the well-known fact that the standard dense graph property testing model, where error $\epsilon$ corresponds to adding or deleting $\epsilon n^2$ potential edges in an $n$-vertex graph, is not well suited to testing
sparse graphs, because any sparse graph trivially has small distance from the empty graph on $n$ vertices.)

Motivated by the above considerations, the relative-error property testing model of \cite{CDHLNSY2024} defines the distance between the unknown $f: \zo^n \to \zo$ and another function $g$ to be
\[
\reldist(f,g) := {\frac {|f^{-1}(1) \hspace{0.05cm} \triangle \hspace{0.05cm} g^{-1}(1)|}{|f^{-1}(1)|}};
\]
so relative distance is measured \emph{at the scale of $|f^{-1}(1)|$} rather than at the ``absolute'' scale of~$2^n =$ $ |\zo^n|$. In addition to the black-box oracle $\MQ(f)$, the model also allows the testing algorithm~to sample i.i.d.~uniform random satisfying assignments of $f$ via a sampling oracle $\SAMP(f)$: each call to 
  $\SAMP(f)$ returns a uniformly random sample from $f^{-1}(1)$ (see \Cref{sec:background-prop-test} and \Cref{rem:relative-versus-dist-free} for a detailed description of the model and comparison with the distribution-free testing model).

The initial work \cite{CDHLNSY2024}
established 
some 
relations between the query complexity of standard-model testing and relative-error testing.  In more detail, letting $q_\mathrm{std}(\eps,n)$ and $q_{\mathrm{rel}}(\eps,n)$ denote the number of oracle calls needed to test a property $\calC$ in the standard model and the relative-error model, respectively, they showed (as Facts 8 and 9 in \cite{CDHLNSY2024}) that 
\begin{equation}\label{eq:lastone}
q_{\mathrm{std}}(\eps,n) \hspace{0.04cm}\lessapprox\hspace{0.04cm}  {q_{\mathrm{rel}}(\eps,n)}\big/{\eps}\ \quad\text{and}\ \quad
q_{\mathrm{rel}}(\eps,n) \hspace{0.04cm} \lessapprox \hspace{0.04cm}q_{\mathrm{std}}(p\eps,n).
\end{equation}
The first inequality shows that standard-model testing is never significantly harder than relative-error testing, but since the value of $p=|f^{-1}(1)|/2^n$ can be extremely small, the second inequality leaves open the possibility that relative-error testing can be much harder than standard-model testing.  Indeed, \cite{CDHLNSY2024} gave an example of an property which is constant-query testable in the standard model but requires \emph{exponentially} many queries for relative-error testing. However, the class of functions exhibiting this separation\footnote{The class consists of all functions $g$ for which  $|g^{-1}(1)|$ is a multiple of $2^{2n/3}$ (see Appendix~A of \cite{CDHLNSY2024}).} is arguably contrived. 
This leads to the following question, which is the motivation for our work:  

\vspace{0.1cm}
\begin{center}
\begin{tcolorbox}[width=13cm]\centering
What is the complexity of relative-error testing \\for \emph{natural and well-studied} properties of Boolean functions?
\end{tcolorbox}
\end{center}\vspace{0.1cm}

  Progress in this direction was made by \cite{CDHLNSY2024}, which give a relative-error tester for monotonicity which makes at most quadratically more queries than the best standard-model tester.\medskip %

\noindent
{\bf The problems we consider.} We attack the above broad question by   studying the relative-error testability of two basic and natural classes of Boolean functions, namely \emph{conjunctions} and \emph{decision lists}.\footnote{Recall that a \emph{conjunction} is an AND of literals $\ell_1 \wedge \cdots \wedge \ell_k$, where each $\ell_i$ is either some Boolean variable $x_j$ or its negation $\overline{x_j}$.
A \emph{decision list} is a sequence of nested ``if (condition) holds, then output (bit), else \dots'' rules, where each (condition) is a Boolean literal (see the beginning of \Cref{sec:DL} for the formal definition of decision lists).} 
Conjunctions and decision lists have been thoroughly studied in both the standard model of Boolean function property testing and the distribution-free model \cite{PRS02,DLM+:07,GlasnerServedio:09toc,CX16,GoldreichRon20,Bshouty20,CFP24}. For the standard model, unlike monotonicity, both problems belong to the regime of ``efficiently testable'' properties, where the number of queries needed depends on $\eps$ only but not the dimension $n$: 
\cite{PRS02,GoldreichRon20} gave $O(1/\eps)$-query algorithms for testing conjunctions~and \cite{Bshouty20} gave an $\tilde{O}(1/\eps)$-query algorithm for testing decision lists (improving on an earlier $\tilde{O}(1/\eps^2)$-query algorithm of \cite{DLM+:07}). In contrast,  distribution-free testing is provably much more difficult: improving on \cite{GlasnerServedio:09toc}, \cite{CX16} gave  an $\tilde{\Omega}(n^{1/3})$ lower bound for conjunctions (and a  matching upper bound);  \cite{CFP24} gave an $\tilde{\Omega}(n^{1/2})$ lower bound for decision lists (and~an $\tilde{O}(n^{11/12})$ upper bound).  

One intuitive reason for this disparity between the standard model and the distribution-free model is that, \emph{under the uniform distribution, (i) every conjunction is either $\eps$-close~to the constant-0 function or else is a $\log(1/\eps)$-junta}, and 
\emph{(ii) every decision list is $\eps$-close to a $\log(1/\eps)$-junta}.
\footnote{A $k$-junta is a function that only depends on at most $k$ of its variables.}  
In the distribution-free setting, though, these facts fail, and notably, they fail in the relative-error setting as well. %
Given this, it is \emph{a priori} unclear whether one should expect the relative-error testability of conjunctions and decision lists to be~more like the standard setting, or more like the distribution-free setting. {Resolving this is one of the motivations of the present work.}

\subsection{Our results and techniques}\label{sec:results-techniques}

{\bf Results.} As our main results, we show that both classes of conjunctions and decision lists are efficiently testable in the 
  relative-error model; indeed we give relative-error
testing algorithms that have the same query complexity as the best standard-model testing algorithms. 

For conjunctions we prove the following theorem:%

\begin{theorem}[Relative-error testing of conjunctions] \label{thm:conjunction}
There is a one-sided non-adaptive
algorithm which is an $\eps$-relative-error tester for conjunctions on $\{0,1\}^n$. The algorithm makes $O(1/\eps)$ calls to the sampling oracle $\SAMP(f)$ and the membership query oracle $\MQ(f)$. 
\end{theorem}

Readers who are familiar with the literature on testing conjunctions in the standard setting   may have noticed that the algorithm above is one-sided and uses $O(1/\eps)$ queries, while only two-sided algorithms are known for testing conjunctions in the standard model
with $O(1/\eps)$ queries \cite{PRS02} and the current best 
one-sided algorithm needs $\tilde{O}(1/\eps)$ queries \cite{GoldreichRon20}. 
It was posed as an open~problem in \cite{Bshouty20}  whether 
  there exists a one-sided, $O(1/\eps)$-query tester for conjunctions in the standard model.
We resolve this question by giving such a tester as a corollary of \Cref{thm:conjunction}:

\begin{theorem}\label{thm:conjunction-standard}
    There is a one-sided, adaptive
algorithm which is an $\eps$-error tester for conjunctions  over $\{0,1\}^n$ in the standard setting. The algorithm makes $O(1/\eps)$ calls to $\MQ(f)$. 
\end{theorem}

Note that combining the first part of \Cref{eq:lastone} and 
  \Cref{thm:conjunction} will only lead to a one-sided upper bound
  of $O(1/\eps^2)$ for the standard model, where the extra $1/\eps$ factor comes from the natural fact that whenever $f$ has density at least $\eps$, it takes $1/\eps$ uniformly random samples to simulate one sample from $f^{-1}(1)$.
To prove \Cref{thm:conjunction-standard}, we obtain a more 
  efficient reduction from relative-error testing to 
  standard testing which, informally, states 
\begin{equation}\label{heeh6}
q_{\textrm{std}}(\eps,n)\lessapprox {1/\epsilon} + q_{\text{rel}}(\eps/2,n),
\end{equation} 
from which \Cref{thm:conjunction-standard} follows.
The reduction is presented in \Cref{sec:rel-error-to-standard}.

Our second, and more involved, main result obtains an efficient relative-error testing algorithm for decision lists:

\begin{theorem} [Relative-error testing of decision lists] \label{thm:DL}
There is an
algorithm which is an  $\eps$-relative-error tester for decision lists on $\{0,1\}^n$. 
The algorithm makes 
$\tilde{O}(1/\eps)$ calls to $\SAMP(f)$ and~$\MQ(f).$
\end{theorem} 

We remark that the algorithm of \Cref{thm:DL} crucially uses the relative-error testing algorithm for conjunctions, i.e.~\Cref{thm:conjunction}, as a subroutine; this is explained more below and a more detailed overview of the algorithm can be found in \Cref{sec:DL-algorithm}.
We also note that the algorithm of \Cref{thm:DL} is adaptive. However, an easy variant of our algorithm and analysis yields a non-adaptive $\epsilon$-relative-error tester for decision lists %
that makes
$O(1/\eps)$ calls to $\SAMP(f)$ and $\tilde{O}(1/\eps^2)$ calls to $\MQ(f)$.

Finally, in \Cref{appendix:lower}, we give a  lower bound of $\Omega(1/\eps)$ for any two-sided, adaptive algorithm for~the relative-error testing of both conjunctions and decision lists, thus showing that  \Cref{thm:conjunction} is optimal and \Cref{thm:DL} is nearly optimal.\medskip\vspace{0.1cm}

{
\noindent {\bf Techniques.}
We give a brief and high-level overview of our techniques, referring the reader to \Cref{sec:conjunction-intuition} and \Cref{sec:DL-algorithm} for more details. Our relative-error testing algorithm for conjunctions begins with a simple reduction to the problem of relative-error testing whether an unknown $f$ is an anti-monotone conjunction\footnote{A conjunction is an anti-monotone conjunction if it only contains negated literals.}.  We solve the relative-error anti-monotone conjunction testing problem with an algorithm which, roughly speaking, first checks whether $f^{-1}(1)$ is close to a linear subspace (as would be the case if $f$ were an anti-monotone conjunction), and then checks whether $f$ is close to anti-monotone (in a suitable sense to meet our requirements).  This is similar, at a high level, to the approach from \cite{PRS02}, but as discussed in \Cref{sec:conjunction-intuition} there are various technical differences; in particular we are able to give a simpler analysis, even for the more general relative-error setting, which achieves one-sided rather than two-sided error.

For decision lists, our starting point is the simple observation that when $f$ is a sparse decision list then $f$ can be written as $C \wedge L$ where $C$ is a conjunction and $L$ is a decision list (intuitively, since $f$ is sparse, a large prefix of the list of output bits of the decision list must all be 0, which corresponds to a conjunction $C$); this decomposition plays an important role for us.  Another simple but useful observation is that all of the variables in the conjunction $C$ must be ``unanimous,'' in the sense that each one of them must always take the same value across all of the positive examples.  Our algorithm draws an initial sample of positive examples from $\SAMP(f)$ and then works with a restriction of $f$ which fixes the variables which were seen to be ``unanimous'' in that initial sample.  Intuitively, if $f$ is a decision list, then it should be the case that the resulting restricted function is close to a \emph{non-sparse} decision list, so we can check that the restricted function is non-sparse, and, if the check passes, apply a standard-model decision list tester to the restricted function.  

This gives some intuition for how to design an algorithm that accepts in the yes-case, but the fact that we must reject all functions that are relative-error far from decision lists leads to significant additional complications. To deal with this, we need to incorporate additional checks %
into the test, and we give an analysis showing that if a function $f$ passes all of these checks then it must be close (in a relative-error sense) to a function of the form $C' \wedge L'$ where $C'$ is relative-error close to a conjunction and $L'$ is relative-error close to a decision list; this lets us conclude that $f$ is relative-error close to a decision list, as desired.  These checks involve applying our relative-error conjunction testing algorithm (which we show to have useful robustness properties against being run with a slightly imperfect sampling oracle --- this robustness is crucial for our analysis); we refer the reader to \Cref{sec:DL-algorithm} for more details.\vspace{0.1cm}
}\medskip

{Looking ahead, while our results resolve the relative-error testability of conjunctions and decision lists, there are many open questions remaining about how relative-error testing compares against standard-error testing for other natural and well-studied Boolean 
function classes, such as $s$-term DNF formulas, $k$-juntas, subclasses of $k$-juntas, unate functions, and beyond.  All of these classes have been intensively studied in the standard testing model, see e.g.~\cite{DLM+:07,CGM11,KhotShinkar16,CS16,LW19,CW19,CWX17stoc,CWX17focs,CW19,Bshouty20}; it is an interesting direction for future work to establish either efficient relative-error testing algorithms, or relative-error testing lower bounds, for these classes.}\vspace{0.1cm}

\noindent{\bf Organization.} \Cref{sec:not-def} contains preliminary notation and definitions, and the general reduction from relative-error testing to standard testing is presented in \Cref{sec:rel-error-to-standard}. We prove \Cref{thm:conjunction} in \Cref{sec:conjunction}, and we prove  \Cref{thm:DL} in \Cref{sec:DL}.

\section{Notation, Definitions, and Preliminaries} \label{sec:not-def}

\subsection{Notation} \label{sec:notation}

For $x,y \in \{0,1\}^n$, we write $y \preceq x$ to indicate that $y_i \leq x_i$ for all $i \in [n]$. 
We write $x \oplus y$ to denote the string in $\{0,1\}^n$ whose $i$-th coordinate is the XOR of $x_i$ and $y_i$.  
Given $S \subseteq [n]$ and $z \in \{0,1\}^n$, we denote by $z_S$ the string in $\{0,1\}^S$ that agrees with $z$ on every coordinate in $S$. 

We use standard notation for  restrictions of functions.  For $f: \zo^n \to \zo$ and $u\in \{0,1\}^S$ for some $S\subseteq [n]$, the function $f{\upharpoonleft_{u}}: \zo^{[n] \setminus S}\rightarrow \{0,1\}$ is defined as 
the function 
\[
f{\upharpoonleft_{u}}(x)=f(z),
\quad \text{where} \quad
z_i = 
\begin{cases}
    u_i & \text{if~}i \in S,\\
    x_i & \text{if~}i \in [n]\setminus S.
\end{cases}
\]
We write ${\cal U}_S$ to denote the uniform distribution over a finite set $S$, and we write ${\cal U}_n$ to denote the uniform distribution over $\zo^n$; we usually skip the subscript $n$ when it is clear from context.

When we refer to vector spaces/subspaces, we will always be talking about the vector space $\F_2^n$, and we will frequently identify $\{0,1\}^n$ with the vector space $\F_2^n$.

Finally, we write $\textsf{DL}$ to denote the class of all decision lists over $\zo^n$, and $\textsf{Conj}$ to denote the class of all conjunctions over $\zo^n$.

\subsection{Background on property testing and relative-error testing} 
\label{sec:background-prop-test}

We recall the basics of the standard property testing model.  In the {\it standard} testing model for a class ${\cal C}$ of $n$-variable Boolean functions, the testing algorithm is given oracle access $\MQ(f)$ to an unknown and arbitrary function $f: \zo^n \to \zo$.
The algorithm must output ``yes'' with high probability (say at least 2/3; this can be amplified using standard techniques) if $f \in {\cal C}$, and must output ``no'' with high probability (again, say at least 2/3) if $\dist(f,{\cal C}) \geq \eps$, where
$$\dist(f,{\cal C}):=\min_{g \in {\cal C}}\hspace{0.05cm}\dist(f,g)\ \quad\text{and}\ \quad 
\dist(f,g) := {\frac {|f^{-1}(1) \ \triangle \ g^{-1}(1)|}{2^n}}.
$$

As defined in \cite{CDHLNSY2024}, a \emph{relative-error} testing algorithm for ${\cal C}$ has oracle access to $\MQ(f)$ and also has access to a $\SAMP(f)$ oracle which, when called, returns a uniform random element $\bx \sim f^{-1}(1)$.
A relative-error testing algorithm for ${\cal C}$ must output ``yes'' with high probability (say at least 2/3; again this can be easily amplified) if $f \in {\cal C}$ and must output ``no'' with high probability (again, say at least 2/3) if $\reldist(f,{\cal C}) \geq \eps$, where
$$\reldist(f,{\cal C}):=\min_{g \in {\cal C}}\hspace{0.05cm}\reldist(f,g)\ \quad\text{and}\ \quad 
\reldist(f,g) := {\frac {|f^{-1}(1) \ \triangle \ g^{-1}(1)|}{|f^{-1}(1)|}}.
$$

Finally, recall that a property testing algorithm for a class of functions ${\cal C}$ (in either the standard model or the relative-error model) is said to be \emph{one-sided} if it returns ``yes'' with probability 1 when the unknown function belongs to ${\cal C}$.

\begin{remark} [Relative-error testing versus distribution-free testing] \label{rem:relative-versus-dist-free} %
Like the distribution-free model, the relative-error model gives the testing algorithm access to both black-box oracle calls to $f$ and i.i.d.~samples from a distribution ${\cal D}$. 
In distribution-free testing ${\cal D}$ is unknown and arbitrary, whereas in relative-error testing ${\cal D}$ is unknown but is guaranteed to be the uniform distribution over $f^{-1}(1)$. 
However, we remark that relative-error testing is \emph{not} a special case of distribution-free testing, because the distance between functions is measured  differently between the two settings. In the distribution-free model, the distance $\Pr_{\bx \sim {\cal D}}[f(\bx) \neq g(\bx)]$ between two functions $f,g$ is measured vis-a-vis the distribution ${\cal D}$ that the testing algorithm can sample from. In contrast, in the relative-error setting the distribution ${\cal D}$ that the tester can sample from is ${\cal D}={\cal U}_{f^{-1}(1)}$, but $\reldist(f,g)$ is \emph{not} equal to $\smash{\Pr_{\bx \sim {\cal U}_{f^{-1}(1)}}[f(\bx)\neq g(\bx)]}$; indeed, a function $g$ can have large relative distance from $f$ because $g$ disagrees with $f$ on many points outside of $f^{-1}(1)$ (which are points that have zero weight under the distribution ${\cal D}$).
\end{remark}

We will need the following simple ``approximate triangle inequality'' for relative distance:

\begin{lemma}[Approximate triangle inequality for relative distance]\label{thm: approx triangle ineq}
    Let $f,g,h:\{0,1\}^n \to \{0,1\}$ be such that $\reldist(f,g)\leq \epsilon$ and $\reldist(g,h)\leq \epsilon'$. Then $\reldist(f,h) \leq \epsilon+(1+\epsilon)\epsilon'$.
\end{lemma}
\begin{proof}
    Since $\reldist(f,g)\leq \epsilon$ we have that $${\frac {|f^{-1}(1) \ \triangle \ g^{-1}(1)|}{|f^{-1}(1)|}} \leq \eps,$$ which means we must have $|g^{-1}(1)| \leq (1+\epsilon)|f^{-1}(1)|$. Hence
    \begin{align*}
        \reldist(f,h) &= {\frac {|f^{-1}(1) \ \triangle \ h^{-1}(1)|}{|f^{-1}(1)|}}\\
        &=\frac{1}{|f^{-1}(1)|}\left| \{ z \::\: f(z) \neq h(z) \}\right| \\
        & \leq \frac{1}{|f^{-1}(1)|}\left| \{ z \::\: f(z) \neq g(z) \lor g(z) \neq h(z) \}\right| \\
        & \leq \frac{1}{|f^{-1}(1)|}\left| \{ z \::\: f(z) \neq g(z) \}\right| + \frac{1}{|f^{-1}(1)|} \left| \{ g(z) \neq h(z) \}\right| \\
        & \leq  \reldist(f,g) + (1+\epsilon) \frac{1}{|g^{-1}(1)|} \left| \{ g(z) \neq h(z) \}\right| \\
        &= \reldist(f,g) + (1+\epsilon) \reldist(g,h) \\
        &\leq \epsilon+(1+{\epsilon})\epsilon'. \qedhere
    \end{align*}
\end{proof}

 Finally, we recall two results for testing decision lists in the standard model that we will use:

\begin{theorem}[\cite{Bshouty20}]\label{thm:Bshouty20}
    There is an adaptive 
algorithm $\mathcal{A}$ for testing decision lists in the standard model with query complexity $\smash{\tilde{O}(1/\eps)}$. In particular, if $f$ is a decision list then $\mathcal{A}$ accepts~with~probability at least $19/20$; if $\dist(f,\sf{DL})\ge\epsilon$ then $\calA$ rejects with probability at least $19/20$. 
\end{theorem}

\begin{theorem}[\cite{DLM+:07}]\label{thm:DLM+:07}
    There is a non-adaptive 
algorithm $\mathcal{B}$ for testing decision lists in the standard model with query complexity $\smash{\tilde{O}(1/\eps^2)}$. In particular, if $f$ is a decision list then $\mathcal{B}$ accepts with probability at least $19/20$; if $\dist(f,\sf{DL})\ge\epsilon$ then $\calB$ rejects with probability at least $19/20$. 
\end{theorem}

\subsection{From relative-error testing to standard testing} \label{sec:rel-error-to-standard}

In this section, we use our relative-error tester for conjunctions (see \Cref{thm:conjunction}) to give an optimal one-sided tester for this class in the standard model. Fact~8 of \cite{CDHLNSY2024} gives a way of transforming any relative-error tester for a property $\mathcal{C}$ into a tester in the standard model for $\mathcal{C}$. However, this transformation incurs a multiplicative $1/\epsilon$ blowup in the number of oracle calls made by the tester. We show that if the all-$0$ function is in $\mathcal{C}$ and the number of oracle calls made by the relative-error tester grows at least linearly in $1/\epsilon$, we can remove this blowup.
\begin{lemma}\label{thm: transormation }
Let $\mathcal{C}$ be a class of Boolean functions over $\zo^n$ that contains  the all-$0$ function. Suppose that there is a relative-error $\epsilon$-testing algorithm $\mathcal{A}$ for $\mathcal{C}$ that makes $q(\epsilon, n)$ many calls to $\MQ(f)$ and $\SAMP(f)$ such that $q( \epsilon/\delta,n) \leq \delta q(\epsilon,n)$ for every $0<\delta<1$. 
Then there is an adaptive standard-model $\epsilon$-testing algorithm $\mathcal{B}$ for $\mathcal{C}$ which makes at most $O(1/\epsilon+q(\epsilon/2, n))$ calls to $\MQ(f)$. Furthermore, if $\mathcal{A}$ is one-sided, so is $\mathcal{B}$. 
\end{lemma}

Before entering into the formal proof, we give some intuition. To prove \Cref{thm: transormation }, we want to run $\mathcal{A}$ on $f$ to see if $f$ is $\epsilon$-far from the class. For this, we will need uniform samples from $f^{-1}(1)$; these can be obtained by randomly querying points of the hypercube to find $1$s of $f$. However if $\Pr_{\bz \sim \mathcal{U}_n}[f(\bz)=1]=:p_f$ is small, then we would need roughly $1/p_f$ queries per sample.
The following key lemma shows that if $f$ is sparse and $\epsilon$-far from $\cal{C}$, then its relative distance to $\mathcal{C}$ is at least $\epsilon/p_f$. This means we can simulate $\mathcal{A}$ on $f$ with the distance parameter for ${\cal A}$ set to a higher value $\epsilon$. Hence while getting samples is hard, this is canceled out by the fact we need less of them.

\begin{lemma}\label{thm: dist to rel-dist}
    If   $\dist(f, \mathcal{C}) > \epsilon$, then $\reldist(f,\mathcal{C}) > \epsilon/p_f$ where $p_f:=|f^{-1}(1)|/2^n$.
\end{lemma}
\begin{proof}
Assume for a contradiction that $\reldist(f,g)\leq \epsilon/p_f$ for some $g\in \calC$. Then we have that:
$${\frac {|f^{-1}(1) \ \triangle \ g^{-1}(1)|}{|f^{-1}(1)|}} = \frac{|f^{-1}(1) \ \triangle \ g^{-1}(1)|}{p_f 2^n } \leq \frac{\epsilon}{p_f},$$
which implies that $|f^{-1}(1) \ \triangle \ g^{-1}(1)| \leq \epsilon 2^n$ and thus, $\dist(f,g)\leq \epsilon$, a contradiction.
\end{proof}

\begin{algorithm}[t!]\caption{Standard-model testing algorithm for $\mathcal{C}$ using a relative-error tester $\mathcal{A}$. Here \\ both $c_1$ and $c_2$ are sufficiently large absolute constants. }\label{algo: transformer}
\vspace{0.15cm}\textbf{Input: } Oracle access to $\MQ(f)$ and an accuracy parameter $\eps > 0$. \\
\textbf{Output: } ``Reject'' or ``Accept.''

\begin{tikzpicture}
\draw [thick,dash dot] (0,1) -- (15.9,1);
\end{tikzpicture}
\begin{algorithmic}[1] \vspace{-0.15cm}
\Algphase{Phase 1: \vspace{-0.1cm} Estimating $|f^{-1}(1)|.$}
    \State Draw $c_1/\epsilon$ many samples $\bz \sim \mathcal{U}_n$ and let $\bS$ be the set of samples.
    \State For each $\bz \in \bS$ query $f(\bz)$ and let $\hat{\bp}$ be the fraction of $\bz\in \bS$ with $f(\bz)=1$.%
    \State If $\hat{\bp} \leq \epsilon/2$, then halt and accept $f$.    
\Algphase{Phase 2: \vspace{-0.1cm} Getting samples}
\State Draw $c_2 \cdot q(\epsilon/2,n)$ many samples $\bz \sim \mathcal{U}_n$ and 
  let $\bG$ be set of samples with $f(\bz)=1$. \State If $\smash{|\bG| \leq q(\epsilon/2 
 \hat{\bp},n)}$, then halt and accept $f$. 
\Algphase{Phase 3: \vspace{-0.1cm} Running the relative error tester}
\State Run $\mathcal{A}$ using the samples in $\bG$ to simulate $\SAMP(f)$ and distance parameter $\epsilon/2 \hat{\bp}$.
\State Output what $\mathcal{A}$ outputs.\vspace{0.15cm}
\end{algorithmic}
\end{algorithm}

We now prove \Cref{thm: transormation }.
\begin{proofof}{\Cref{thm: transormation }}
    Our goal will be to show that \Cref{algo: transformer} is a standard testing algorithm for $\mathcal{C}$. Without loss of generality, we assume that $\mathcal{A}$ is correct with high enough probability, that is, if $\reldist(f,\mathcal{C}) \geq \epsilon$ the algorithm rejects with probability at least $0.9$. Similarly, if $f \in \mathcal{C}$, we assume that $\mathcal{A}$ accepts with probability at least $0.9$. 

    First, consider the case where $f \in \mathcal{C}$. Note that $f$ can only be rejected if $\mathcal{A}$ rejects $f$. Since $\bG$ is a set of uniformly random samples from $f^{-1}(1)$, $\mathcal{A}$ must accept with probability at least $0.9$.~And if $\mathcal{A}$ is one-sided, then $f$ is always accepted, meaning that   \Cref{algo: transformer} is also one-sided. 

    Now, assume that $f$ is $\epsilon$-far from any function in $\mathcal{C}$. Because the all-$0$ function is in $\mathcal{C}$ we must have $p_f:= {|f^{-1}(1)|}/{2^n}\geq \epsilon$. So on line 2, we have by a Chernoff bound that $$\Pr\big[|\hat{\bp}-p_f|>0.05p_f\big] \leq 2e^{\frac{-(0.05)^2 \epsilon}{2+0.05} c_1/\epsilon} \leq 0.05,$$ by picking a suitably large constant $c_1$. Below we assume that $\hat{\bp} \in (1 \pm 0.05)p_f$. 
    
    Since $p_f \geq \epsilon$, this implies $\hat{\bp} \geq \epsilon/2$. So, the algorithm doesn't accept on line $3$. We also~have that $\Pr_{\bz \sim \mathcal{U}_n}[f(\bz)=1]\geq 0.9 \hat{\bp}$. Hence, the expected size of  $|\bG|$ is at least $0.9c_2\hat{\bp} \cdot q(\epsilon/2,n)$. By our assumption on $q(\cdot ,\cdot)$ this is at least $0.9c_2\cdot q(\epsilon/2\hat{p}, n)$.
    Again, by a Chernoff bound and by picking
      $c_2$ to be a sufficiently large constant, we have that  $|\bG|\leq q(\epsilon/2 \hat{\bp},n)$ with probability at most $0.05$.

 It remains to argue that $\mathcal{A}$ rejects with high probability on line $6$. By \Cref{thm: dist to rel-dist}, we have~that $\reldist(f,\mathcal{C})\geq \epsilon/p_f$. Using $\hat{\bp} \geq 0.95p_f$, we have that $\reldist(f,\mathcal{C}) \geq \epsilon/2\hat{\bp}$. Since we simulate $\mathcal{A}$ on line $6$ using uniform samples from $f^{-1}(1)$, it must accept $f$ with probability at most $0.1$.~As~a~result, the probability the algorithm accepts $f$ is at most $0.05+0.05+0.1 \leq 0.2$.

     Line 2 uses $O(1/\epsilon)$ many queries, line 4 uses $O(q(\epsilon/2,n))$ many queries, and simulating $\mathcal{A}$ on line 6 uses $q(\epsilon/2\hat{\bp}, n) = O(q(\epsilon/2,n))$ queries. So, the overall query complexity  is $O(1/\epsilon+q(\epsilon/2,n))$.
\end{proofof}

Using the above, we can obtain a one-sided tester for conjunctions in the standard setting from our one-sided relative-error conjunction tester. The result of \cite{BshoutyGoldreich022} implies an $\Omega(1/\epsilon)$ lower bound for this problem, hence this result is tight. 

\begin{proofof}{\Cref{thm:conjunction-standard}}
    Our relative-error tester for conjunctions in \Cref{thm:conjunction} makes $O(1/\eps)$ calls to $\SAMP(f)$~and $\MQ(f)$  and is one-sided. As the all-$0$ function is a conjunction, the theorem follows by \Cref{thm: transormation }. 
\end{proofof}

\section{A relative-error testing algorithm for conjunctions} \label{sec:conjunction}

In this section we prove \Cref{thm:conjunction}. {At the end of the section, we show that the main  algorithm for conjunctions works even when 
  its access to the sampling oracle is not perfect; this will be important when the algorithm is used as a subroutine to test decision lists in the next section.}

\begin{remark} \label{remark:monotone}
In the rest of this section, we give a relative-error tester for the class of \emph{anti-monotone} conjunctions. That is, for conjunctions of {\it negated} variables. Using an observation from \cite{PRS02} (see also \cite{GoldreichRon20}), this implies a relative-error tester for general conjunctions. 
Specifically, a slightly modified version of Observation 3 from Parnas et al.  \cite{PRS02} states that for any conjunction $f$ (not necessarily anti-monotone) and for any $y \in f^{-1}(1)$, the function
\[f_{y}(x) :=f(x \oplus y)\] 
is an anti-monotone conjunction. Furthermore, it is easy to see that if $f$  is $\epsilon$-relative-error far~from every conjunction, then $f_{{y}}$ must be $\epsilon$-relative-error far from every anti-monotone conjunction. %
So to get a relative-error tester for general conjunctions, we simply draw a single $y \sim f^{-1}(1)$ and then run the tester for anti-monotone conjunctions (\Cref{algo: mono conjunction tester}) on $f_{y}$, {where $\MQ(f_y)$ and $\SAMP(f_y)$ can be simulated  using 
 $\MQ(f)$ and $\SAMP(f)$, respectively}; we refer to this algorithm as \hypertarget{Algorithm2}{\sc Conj-Test}. %
\end{remark}

\subsection{Preliminaries}
\begin{definition}[Subspaces]  \label{def:affine-subspace}
A set $W \subseteq \{0,1\}^n$ is an affine subspace over $F_2^n$ if and only if $W$ is the set of solutions to a system of linear equations over $F_2^n$. That is, iff $W$ is the set of solutions to $Ax=b$, for some $b \in \{0,1\}^n$. 
$W$ is a linear subspace  over $F_2^n$ iff it is the set of solutions to $Ax=b$ where $b$ is the all-0 vector.
A well-known equivalent characterization states that $W$ is an affine subspace if and only if for all $x,y,z \in W$, $x \oplus y \oplus z \in W$.
And for linear subspaces, we have that $W$ is linear if and only if for all $x,y \in W$, $x \oplus y \in W$.
\end{definition}

We will need the following fact about the intersection of linear subspaces. %

\begin{fact}\label{claim:PRS-15}
    Let $H, H'$ be two linear subspaces of $\{0,1\}^n$ with $H \not \subseteq H'$. Then: 
    \[\frac{|H \cap H'|}{ |H| } \leq \frac{1}{2}.\]
\end{fact}

\begin{definition}[Anti-monotone set]
A set $H \subseteq \{0,1\}^n$ is said to be \emph{anti-monotone} if and only if whenever $x \in H$, we have $y \in H$ for all $y \preceq x$. 
\end{definition}

We record the following simple result about anti-monotone conjunctions, which will be useful for understanding and analyzing \Cref{algo: mono conjunction tester}:

 \begin{lemma}\label{lem: monotone conj characterization}
     A function $f: \zo^n \to \zo$ is an anti-monotone conjunction if and only if $f^{-1}(1)$ is anti-monotone and is a linear subspace of $\zo^n$. 
 \end{lemma}

\begin{proof}
The ``if'' direction follows immediately from the following easy result of \cite{PRS02}:
\begin{claim} [Claim~18 of \cite{PRS02}]
    If $H$ is a monotone and affine subspace of $\{0,1\}^n$, then $H=\{x \in \zo^n : (x_{i_1}=1) \land \ldots \land (x_{i_k} =1) \}$ for some $k$ and some $i_1,\dots,i_k \in [n]$. 
    Similarly if $H$ is an anti-monotone and linear subspace, then 
    $H = \{x \in \{0,1\}^n ~:~(x_{i_1}=0) \land \ldots \land (x_{i_k}=0)\}$ for some $k$ and some $i_1,\ldots,i_k \in [n]$.
\end{claim}

For the other direction, suppose that $f(x)=x_{i_1} \land \cdots \land x_{i_k}$ is an anti-monotone conjunction. Then clearly $f^{-1}(1)$ is an anti-monotone set. Moreover, for $x,y \in f^{-1}(1)$, 
$f(x \oplus y) =  \neg (x_{i_1} \oplus y_{i_1}) \land \ldots \land \neg (x_{i_k} \oplus y_{i_k})$.
Since $x,y$ are both in $f^{-1}(1)$, $x_{i_1},...,x_{i_k},y_{i_1},\ldots,y_{i_k}$ are all zero, and thus $f(x \oplus y)=1$.
Recalling \Cref{def:affine-subspace}, we have that $f^{-1}(1)$ is a linear subspace, so the only-if direction also holds.
\end{proof}
 
\subsection{Intuition: the tester of \cite{PRS02} and our tester}
\label{sec:conjunction-intuition}

We begin with an overview of our algorithm and its proof.
At a bird's eye view, our algorithm~has two phases: Phase 1  tests whether $f^{-1}(1)$ is close to a linear subspace, and Phase 2 tests whether $f$ is anti-monotone.  
(Note that, while \cite{CDHLNSY2024} already gave a general tester for monotonicity in the relative error setting, its complexity has a dependence on $2^n/|f^{-1}(1)|$ which is necessary in the general case.  Here we want a much more efficient tester that is linear in $1/\epsilon$, so we cannot use the monotonicity tester of \cite{CDHLNSY2024}.)

If $f$ is an anti-monotone conjunction then both Phase 1 and Phase 2 tests will clearly pass; the hard part is to prove soundness, showing that when $f$ is far from every anti-monotone conjunction, then the algorithm will reject with high probability.
The crux of the analysis is to show that if $f$ is close to a linear subspace but far from all anti-monotone conjunctions, then the Phase 2 anti-monotonicity test will fail with high probability.
As in \cite{PRS02}, we accomplish this with the aid of an intermediate function, $g_f$ that we will use in the analysis.
In Phase 1 of the algorithm,  we will actually prove something stronger: if $f$ passes Phase 1 with good probability, we show that $f$ is close to a particular function, $g_f$ (that is defined based on $f$), where $g_f^{-1}(1)$ is a linear subspace. 
The function $g_f$ is chosen to help us downstream in the analysis, to argue that if Phase 2 passes, then $g_f$ is anti-monotone  with good probability; since $g_f^{-1}(1)$ is a linear subspace, by \Cref{lem: monotone conj characterization} this implies $g_f$ is an anti-monotone conjunction,
thereby contradicting the assumption that $f$ is far from every anti-monotone conjunction.

Finally, we remark that as described in \Cref{sec:rel-error-to-standard}, our relative-error tester yields an $O(1/\eps)$-query one-sided tester in the standard model.  This strengthened the result~of~\cite{PRS02}, which~gave an $O(1/\eps)$-query algorithm with two-sided error. Comparing our analysis with \cite{PRS02},  while~we borrow many of their key ideas, our analysis is arguably simpler in a couple of ways.   First since we are forced to create a tester that doesn't have the ability to estimate the size of the conjunction, this eliminates the initial phase in the \cite{PRS02} which does this estimation, giving a more streamlined argument {and enabling us to give an $O(1/\eps)$-query algorithm with one-sided error.} %
Second, by a simple transformation (see  \Cref{remark:monotone}), we test whether $f$ is close to a monotone {\it linear} subspace rather than close to a monotone affine subspace, which makes both the test and analysis easier.

\subsection{The relative-error anti-monotone conjunction tester. }
\label{sec:conj-algorithm}

\begin{algorithm}[t!] 
\caption{Relative-error Anti-monotone Conjunction Tester. 
{Here $c_1$ and $c_2$ are two\\ sufficiently large absolute constants. } }

\label{algo: mono conjunction tester}

\vspace{0.15cm}\textbf{Input: } Oracle access to $\MQ(f), \SAMP(f)$ and an accuracy parameter $\eps > 0$. \\
\textbf{Output: } ``Reject'' or ``Accept.''

\begin{tikzpicture}
\draw [thick,dash dot] (0,1) -- (15.9,1);
\end{tikzpicture}
\begin{algorithmic}[1] \vspace{-0.15cm}
\Algphase{Phase 1: \vspace{-0.1cm}}
    \State Repeat the following $c_1/\epsilon$ times :
    \State \hskip4em Draw $\bx,\by \sim \SAMP(f)$.
    \State \hskip4em If $\MQ(f)(\bx \oplus \by)=0$, halt and reject $f$.\vspace{-0.15cm}

\Algphase{Phase 2:\vspace{-0.15cm}}
\State Repeat the following $c_2$ times :
    \State \hskip4em Draw $\bx \sim \SAMP(f)$.
    \State \hskip4em Draw a uniform random $\by \preceq \bx$.
    \State \hskip4em Draw $\bu \sim \SAMP(f)$.
    \State \hskip4em If $\MQ(f)(\by \oplus \bu)=0$, halt and reject $f$.

\State Accept $f$.\vspace{0.15cm}
\end{algorithmic} 
\end{algorithm}

In the rest of this section we prove the following theorem, which implies 
\Cref{thm:conjunction} via \Cref{remark:monotone}:

\begin{theorem} [Relative-error anti- monotone conjunction tester] \label{thm:acceptable-monotone}
\Cref{algo: mono conjunction tester} is a one-sided non-adaptive
algorithm which, given $\MQ(f)$ and $\SAMP(f)$,
is an $\eps$-relative-error tester for anti-monotone conjunctions over $\{0,1\}^n$, making $O(1/\eps)$ oracle calls.
\end{theorem}

The $O(1/\eps)$ bound on the number of oracle calls is immediate from inspection of \Cref{algo: mono conjunction tester}, so it remains to establish correctness (including soundness and completeness).
These are established in \Cref{thm: f is mono conj} and \Cref{thm: f isn't mono conj} respectively.

\subsection{Completeness} %

Completeness (with one-sided error) is straightforward to establish:  
\begin{theorem}\label{thm: f is mono conj}
If $f: \zo^n \to \zo$ is any anti-monotone conjunction, then \Cref{algo: mono conjunction tester} accepts $f$ with probability 1.
\end{theorem}

\begin{proof}
The algorithm can only reject $f$ on lines $3$ or $8$. We argue that if $f$ is an anti-monotone conjunction then the algorithm never rejects at either of these lines. Since $f$ is an anti-monotone conjunction, by \Cref{lem: monotone conj characterization} $f^{-1}(1)$ is a linear subspace. Hence for any $x,y \in f^{-1}(1)$ we also have $f(x \oplus y)=1$, so the algorithm never rejects at line $3$.
Continuing to Phase 2, since $f$ is anti-monotone, $\bx \in f^{-1}(1)$, and $\by \preceq \bx$, we also have $f(\by)=1$. Thus, since $\by, \bu \in f^{-1}(1)$ and $f^{-1}(1)$ is a linear subspace, recalling \Cref{def:affine-subspace}, the algorithm never rejects at line $8$. 
\end{proof}

\subsection{Soundness} \label{sec:conjunction-soundness}
Let $g_f : \{0,1\}^n \to \{0,1\}$ be the function defined as follows:
\begin{equation} \label{eq:g}
g_f(a)=\begin{cases}
         1 & \text{if $\Pr_{\bx \sim f^{-1}(1)}[f(a \oplus \bx)=1] \geq 1/2$} \\
         0 & \text{otherwise}
    \end{cases}
\end{equation}
For convenience, we will write $F := f^{-1}(1)$, $G := g_f^{-1}(1)$ and $N:=|F |$. %

\subsubsection{Analysis of phase 1 (steps 1-3)}

We want to prove the following two lemmas which together say that whenever \Cref{algo: mono conjunction tester} passes Phase 1 with probability at least $1/10$, then  (1) $f$ is  relative-error-close to $g_f$ (\Cref{lem:f is close to g conj}) which means that $|F \triangle G |$ is small relative to $N$, and (2) $G $ is a linear subspace (\Cref{lem:g is affine}).

\begin{lemma}\label{lem:f is close to g conj}
Suppose that when \Cref{algo: mono conjunction tester} is run on $f:\{0,1\}^n\rightarrow \{0,1\}$, it passes Phase 1 (that is, it reaches  line $4$) with probability at least $1/10$. Then $\reldist(f,g_f) \leq \epsilon/10$.
\end{lemma}

\begin{lemma}\label{lem:g is affine}
Suppose that when \Cref{algo: mono conjunction tester} is run on $f:\zo^n \to \zo$, it passes Phase 1 (that is, it reaches line ${4}$) with probability at least $1/10$.
 Then $G $ is a linear subspace.
\end{lemma}

\begin{proofof}{\Cref{lem:f is close to g conj}}
Assume for a contradiction that  
 $\reldist(f,g_f) \geq \epsilon/10$.
 First we show that  
\begin{equation}
\label{eq:hehe10}\Prx_{\bx,\by \sim F } \big[f(\bx \oplus \by) =0\big] \geq \epsilon/40.\end{equation}
Therefore,  the probability we fail to reject $f$ during Phase 1 is at most $(1-\epsilon/40)^{c_1/\epsilon}$, which is at most $1/10$ for a suitable choice of $c_1$. 

It remains to prove \Cref{eq:hehe10}. %
Since $\reldist(f,g_f) \geq \epsilon/10$, we have $|F \Delta G |/N \geq \epsilon/10$ and 
there are two cases: at least half of the elements in $F \Delta G $ are in $F $ but not in $G $, or vice versa. %

In the first case (at least half of the elements in $F \Delta G $ are in $F $ but not in $G $),  we have:
$$\Prx_{\bx,\by \sim F }\big[f(\bx \oplus \by)=0\big] \geq 
\Prx_{\bx, \by \sim F } \big[g_f(\by)=0\big] \cdot \Prx_{\bx,\by \sim F}\big[f(\bx \oplus \by)=0 \hspace{0.08cm}|\hspace{0.08cm} g_f(\by)=0\big].$$
Since we are in the first case, the first term is at least $\epsilon/20$, and for the second term, because we are conditioning on  $g_f(\by)=0$, by the definition of $g_f$, this term is at least $1/2$. \Cref{eq:hehe10} follows.

For the second case, we have $|G \setminus F | \geq \eps N / 20$. Fix any $a \in G \setminus F $, since $g_f(a)=1$, we have $\Pr_{\bx \sim F}[f(a \oplus \bx )=1] \geq 1/2.$
For any pair $(a,x)$ such that $f(a\oplus x )=1$, we have $a \oplus x =y$ for some $y \in F$. As a result, we have
\[
\Prx_{\bx,\by \sim F}\big[a \oplus \bx  = \by\big] \geq {\frac 1 2} \cdot {\frac 1 N}\quad\Longrightarrow\quad
\Prx_{\bx,\by \sim F}\big[\bx \oplus \by = a \big] \geq \frac{1}{2N}.
\]
Since we are in the second case, this gives
\[
\Prx_{\bx,\by \sim F}\big[f(\bx \oplus \by )=0 \big]
\ge \Prx_{\bx,\by \sim F}\big[ \bx \oplus \by\in G\setminus F\big]
=\sum_{a \in G \setminus F }\Prx_{\bx,\by \sim F}\big[\bx \oplus \by = a \big] \geq {\frac {\eps N} {20}}\cdot {\frac 1 {2N}}  = \frac{\epsilon}{40},
\]
which completes the proof of the lemma.
\end{proofof}
Recalling \Cref{def:affine-subspace}, we will prove \Cref{lem:g is affine} by showing that under the stated condition on $f$,
\text{for any $a,b \in G$, we have $g_f(a \oplus b)=1$}.

\def\W{\mathsf{W}}
Following \cite{PRS02},  we define the set of \emph{witnesses}  for $a \in \{0,1\}^n$ to be
$$
\W(a):=\big\{x \in F: f(a\oplus x ) \neq f(a)\big\}.
$$
Elements of the set $\W(a)$ are witnesses in the sense that they witness that $F$ is not a linear subspace. 
Let 
$$H = \big\{a \in \{0,1\}^n :  |\W(a) | > \delta N\big\}$$
be the set of elements $a \in \zo^n$ for which the set of witnesses for $a$ is ``large,'' i.e. at least a $\delta$ fraction of $N=|F|,$ where $\delta:=1/36$.
The next lemma states that if the algorithm passes Phase 1 with probability at least $1/10$, then the set $H$ is small. 

\begin{lemma} \label{lem: H big we reject}
    Assume that $|H| \geq \delta  N$. 
    Then $\Pr\big[\text{\Cref{algo: mono conjunction tester} on $f$ reaches line 7}\big] < 1/10.$
\end{lemma}

\begin{proof}
Suppose that $|H| \geq \delta  N$. Similar to the proof of  \Cref{lem:f is close to g conj}, we consider two cases: either (i)  $|H \cap F| \geq \delta N /2$, or
(ii) $|H \cap f^{-1}(0)| \geq \delta N /2$.
We will argue that in both cases we have
\begin{equation} \label{eq:xyz}
\Prx_{\bx,\by \sim F}\big[f(\bx \oplus \by)=0\big] \geq \delta^2/2
\end{equation}
and thus, the probability that we fail to reject $f$ in line $3$ is at most $(1-\delta^2/2)^{c_1/\epsilon}$, which is less than $1/10$ for a suitably large $c_1$.

In the first case, we have for any $y \in H \cap F$:
$$\Prx_{\bx \sim F}\big[f(y \oplus \bx)=0\big] \geq {\delta}.
$$ 
On the other hand in this first case $\by \sim F$ belongs to $H \cap F$ with probability at least ${\delta/2}$. So altogether we have 
$$\Prx_{\bx,\by \sim F}\big[f(\bx \oplus \by)=0\big] \geq \delta^2/2,$$
giving~\Cref{eq:xyz} as desired in the first case.
    
In the second case, there exist at least $\delta N /2$ points $a$ with $f(a)=0$ and $a \in H$. For any such point $a \in H \cap f^{-1}(0)$ we have
    $$\Prx_{\bx \sim F}\big[f(a \oplus \bx)=1\big] \geq \delta.$$
    But if $f(a \oplus x)=1$ then $a \oplus x = y$ for some $y \in F$. Hence
    $$\Prx_{\bx, \by \sim F}\big[a \oplus \bx  = \by \big] \geq \delta/N,
\text{~which we can rewrite as~}
\Prx_{\bx,\by \sim F}\big[\bx \oplus \by = a \big] \geq \delta/N.$$
    Summing over all $a \in H \cap f^{-1}(0)$, we have 
\begin{align*}
    \Prx_{\bx,\by \sim F}\big[f(\bx \oplus \by)=0\big]
     &\ge \Prx_{\bx,\by \sim F}\big[ \bx \oplus \by\in H\cap f^{-1}(0)\big]\ge \frac{\delta N}{2}\cdot \frac{\delta}{N}\ge \frac{\delta^2}{2},\end{align*}
so indeed \Cref{eq:xyz} holds in the second case as well.
    \end{proof}

We know from the definition of $g_f$ that for every $a \in \{0,1\}^n$, $$\Prx_{\bx \sim F} \big[g_f(a)=f(a \oplus \bx)\big] \geq 1/2.$$
The next Lemma states that whenever $|H|$ is small, %
the agreement probability is a lot higher. This in turn will allow us to prove \Cref{lem:g is affine}. 

\begin{lemma} 
\label{lem:amplification}
    Assume that $|H| \leq \delta N$. Then for every $a \in \{0,1\}^n$,  we have  
    $$\Prx_{\bx \sim F}\big[g_f(a) = f(a \oplus \bx)\big] \geq 1- 4 \delta.$$
\end{lemma}

Armed with the above Lemma (which we defer to the end of this section), we can now complete the proof of  \Cref{lem:g is affine}.

\begin{proofof}{\Cref{lem:g is affine}}
By \Cref{lem: H big we reject} if $|H| >\delta N$, then with probability at least $9/10$, $f$ is rejected in Phase 1. 
Thus it suffices to show that if $|H| \leq \delta N$, then $G$ is a linear subspace.
Towards this goal, assume for sake of contradiction that $G$ is not a linear subspace, and thus there exist $a,b \in \{0,1\}^n$ such that $g_f(a)=g_f(b)=1$ but $g_f(a \oplus b)=0$.
At a high level, we obtain a contradiction as follows: 
\begin{itemize}
    \item [(1)] First, using the assumption that $g_f(a)=g_f(b)=1$ plus  \Cref{lem:amplification} we show that
$$\Prx_{\bx,\by \sim F}\big[f(a \oplus \bx \oplus b \oplus \by)=1\big] \geq 1-10\delta.$$
\item[(2)] On the other hand using the assumption that $g_f(a \oplus b)=0$ plus \Cref{lem:amplification} we show that
$$\Prx_{\bx,\by \sim F}\big[f(a \oplus \bx \oplus b \oplus \by)=1\big] < 6 \delta.$$
\end{itemize}
Together this gives a contradiction by our choice of $\delta=1/36$.

\medskip

\noindent {\it Proving (1):}
First, by definition of $W(\cdot)$ we have
\begin{align*}
    \Prx_{\bx,\by \sim F}\big[f(a \oplus \bx \oplus b \oplus \by)=1\big]   &\geq \Prx_{\bx,\by \sim F}\left[\big( a \oplus \bx \in F \setminus H\big) \land \big(b \oplus \by \in F \setminus W(a \oplus \bx)\big)\right] \\[0.8ex]
    & \geq \Prx_{\bx \sim F}\big[ a \oplus \bx \in F \setminus H\big] \times \min_{a' \in F \setminus H} \Prx_{\by \sim F}\big[ b \oplus y \in F  \setminus W(a')\big].
\end{align*}
We claim that the first term above is at least $1-5 \delta$.
This is because 
\begin{flushleft}\begin{enumerate}
\item $\Pr_{\bx \sim F}[a \oplus \bx \in F] \geq 1-4\delta$ using $g_f(a)=1$ and \Cref{lem:amplification}, and 
\item $\Pr_{\bx \sim F}[a \oplus \bx \in H] \leq \delta$ since $|H|\le \delta N$ and for a fixed $a$, by linearity every $z \in H$ has {at most one} unique representation as $a \oplus \bx$ with $\bx\in F$. %
\end{enumerate}\end{flushleft}
It follows that the first term is at least $1-5\delta$.

Next we claim that the second term is at least $1-5 \delta$.
This is because
\begin{flushleft}\begin{enumerate}
    \item 
 $\Pr_{\by \in F}[b \oplus \by \in F ] \geq 1-4 \delta$ using $g_f(b)=1$
 and \Cref{lem:amplification}, and 
\item Given that $a' \in F \setminus H$, we have  
  $|W(a')| \leq \delta N$ and thus,
 $\Pr_{\by \sim F}[b \oplus \by \in W(a')] \leq \delta$. 
\end{enumerate}\end{flushleft}
So the second term is at least $1-5 \delta$, and 
  the product is at least $(1-5\delta)(1-5\delta) \geq 1-10\delta$.
\medskip

\noindent {\it Proving (2):}
We bound the probability of interest as follows:
\begin{align*}
    &\hspace{-1cm}\Prx_{\bx,\by \sim F}\big[f(a \oplus \bx \oplus b \oplus \by)=1\big] \\
    &\leq \Prx_{\bx \sim F}\big[a\oplus b \oplus \bx \in H\big] + \Prx_{\bx \sim F}\big[f(a \oplus b \oplus \bx)=1\big] \\
    & ~~~~~~~~~+\Prx_{\bx,\by \sim F}\big[f(a\oplus \bx \oplus b \oplus \by)=1 \hspace{0.08cm}|\hspace{0.08cm} a\oplus b \oplus \bx \not \in H \land f(a \oplus b \oplus \bx)=0\big].
\end{align*}
The first term is at most $\delta$ by linearity;
the second term is at most $4\delta$ since $g_f(a\oplus b)=0$ and by \Cref{lem:amplification}, and the last term is also bounded by $\delta$ by the definition of $H$. 
Therefore the probability that $f(a \oplus \bx \oplus b \oplus \by)=1$ is at most $6 \delta$. 
\end{proofof}

It is left to prove \Cref{lem:amplification}.

\begin{proofof}{\Cref{lem:amplification}}
We want to prove that if $H$ is small ($|H| \leq \delta N$), then the agreement between $g_f$ and $f$ is amplified in the sense that for any $a\in \{0,1\}^n$,  $\Pr_{\bx \sim F}[g_f(a)=f(a \oplus \bx)] \geq 1-4\delta$.~By the definition of $g_f$, this probability is at least $1/2$. Clearly any $a \not \in H$ has at most $\delta N$ witnesses $x$ that violate the equality, so for these $a$'s, the probability is easily seen to get amplified to $1-\delta \geq 1-4\delta$. However, we want to show that small $H$ implies that the agreement between $g_f$ and $f$ is amplified for {\it every} $a$. 
To show this, fix $a$ and let $\gamma = \Pr_{\bx \sim F}[f(a \oplus \bx) = g_f(a)]$.
We will consider the quantity
\begin{equation}\label{amplif-eqn}
\Prx_{\bx,\by \sim F}\big[f(a \oplus \bx) = f(a \oplus \by)\big].
\end{equation}
On the one hand, we will relate it to $\gamma$ and on the other hand derive an expression in terms of $\delta$.
\begin{align*}
         &\hspace{-1 cm}\Prx_{\bx,\by \sim F}\big[f(a \oplus \bx) = f(a \oplus \by)\big]
         \\&\geq \Prx_{\bx,\by \sim F}\big[f(a \oplus \bx)=f(a \oplus \bx \oplus \by) \land  f(a \oplus \by)=f(a \oplus \bx \oplus \by )\big]  \\
        &= 1 - \Prx_{\bx,\by \sim F}\big[f(a \oplus \bx) \neq f(a \oplus \bx \oplus \by) \lor  f(a \oplus \by) \neq f(a \oplus \bx \oplus \by)\big] \\
        & \geq 1 - 2 \times \Prx_{\bx, \by \sim F} \big[f(a \oplus \bx) \neq f(a \oplus \bx \oplus \by)\big].
\end{align*}
Now we will show that the last term is at most $2 \delta$:
\begin{align*}
        &\hspace{-1 cm}\Prx_{\bx, \by \sim F}\big[f(a \oplus \bx) \neq f(a \oplus \bx \oplus \by)\big]\\[-1ex]
        &\leq \Prx_{\bx \sim F}\big[a \oplus \bx \in H\big] \times \max_{a' \in H} \left\{ \Prx_{\by \sim F}\big[f(a') \neq f(a' \oplus \by)\big] \right\} \\ 
        &\hspace{1cm}+ \Prx_{\bx \sim F}\big[a \oplus \bx \not \in H\big] \times \max_{a' \not \in H} \left\{\Prx_{\by \sim F}\big[f(a') \neq f(a' \oplus \by)\big]\right\}.
    \end{align*}
The second term is at most $\delta$ since for every $a' \not \in H$, $\Pr_{\by \sim F}[f(a') \neq f(a' \oplus \by)] \leq \delta$ and trivially~we have $\Pr_{\bx \sim F}[a \oplus \bx \not \in H] \leq 1$.
The first term is also at most $\delta$, since we have $\Pr_{\bx \sim F}[a \oplus \bx \in H] \leq \delta$  by linearity %
and trivially $\Pr_{\by \sim F}[f(a') \neq f(a' \oplus \by)] \leq 1$. 
Therefore, \Cref{amplif-eqn} is at least $1-4 \delta$.

Now we will express the same quantity in terms of the agreement probability, $\gamma$:
\begin{align*}
        &\hspace{-1cm}\Prx_{\bx, \by \sim F}\big[f(a \oplus \bx) = f(a \oplus \by)\big] \\[-0.5ex]
        &= \Prx_{\bx, \by \sim F}\big[f(a \oplus \bx)=g_f(a)  \land  f(a \oplus \by)=g_f(a)\big] \\[-0.5ex]
        &\hspace{1cm}+ \Prx_{\bx,\by \sim F}\big[f(a \oplus \bx) \neq g_f(a) \land  f(a \oplus \by) \neq g_f(a)\big] \\
        &= \gamma^2 + (1-\gamma)^2,
\end{align*}
using independence of $\bx$ and $\by$.
By combining our two inequalities,  we have 
$\gamma^2+(1-\gamma)^2 \geq 1 - 4 \delta$. Rearranging we have $\gamma(1-\gamma) \leq 2\delta$. But $\gamma \geq 1/2$, so $1 - \gamma \leq 4 \delta$, as we wanted to show.
\end{proofof}

\subsubsection{Analysis of phase 2 (lines 4-9)}

For the second phase, we will establish the following lemma that states that any $f$ that has passed Phase 1 but is far from an anti-monotone conjunction, will be rejected with high probability.

\begin{lemma}\label{lem: line 7-11}
    Consider an execution of \Cref{algo: mono conjunction tester}  up through line 8 on a function $f$  that satisfies the following conditions:
    \begin{itemize}
        \item $f$ is $\epsilon$-far in relative distance from every {anti-monotone} conjunction;
        \item $\reldist(f,g_f) \leq \epsilon/10$;
        \item and $G:=g_f^{-1}(1)$ is a linear subspace.
    \end{itemize}
Then $\Pr[$\Cref{algo: mono conjunction tester} rejects $f$ in lines 4-8$]\geq 0.9$.
\end{lemma}

\def\LT{\textsf{LT}}

We will need the following two definitions. For any $x \in G$, we define 
\begin{align*}
\LT(x) =\big\{ y \in \zo^n : y \preceq x \big\}\quad\text{and}\quad
\mathcal{X} =\big\{x \in G : \LT(x) \subseteq G\big\}.
\end{align*}

We recall the following useful claim that is analogous to Claim~21 from \cite{PRS02}:%

\begin{claim} \label{lemma:PRS02} %
If $G$ is a {linear} subspace of $\zo^n$ then the set $\mathcal{X}$ is a {linear} subspace of $G$. Moreover, if $g_f$ is not an anti-monotone conjunction then $|\mathcal{X}|\leq  |G|/2$.
\end{claim} 
\begin{proof}
    We first show $\mathcal{X}$ is a linear subspace of $G$. To do this, we need to show that if $x^1,x^2 \in \mathcal{X}$ then $z=x^1 \oplus x^2 \in \mathcal{X}$, meaning $y \in G$ for every $y \preceq z$. Fix $y\preceq z$, we want to write $y=y^1 \oplus y^2$ where $y^1 \preceq x^1, y^2 \preceq x^2$. For a coordinate $i$ if $y_i=0$, then we set $y^1_i=y^2_i=0$. Otherwise if $y_i=1$, then $z_i=1$ so at least one of $x^1_i$ or $x^2_i$ is $1$. If $x^1_i=1$, we set $y^1_i=1$ and $y^2_i=0$ and otherwise  $y^2_i=1$ and $y^1_i=0$. It's easy to check $y^1 \preceq x^1$ and $y^2 \preceq x^2$. Since $x^1,x^2 \in \mathcal{X}$ this implies $y^1,y^2 \in G$ and since $G$ is a linear subspace we also have $y \in G$. So $\mathcal{X} \subseteq G$. 

    If $\mathcal{X}=G$, then $G$ is an anti-monotone subspace. By \Cref{lem: monotone conj characterization}, this implies $g_f$ is an anti-monotone conjunction. So, if $g_f$ isn't an anti-monotone conjunction, we must have $G \not \subseteq \mathcal{X}$, by \Cref{claim:PRS-15} this implies $|\mathcal{X}| \leq |G|/2$. 
\end{proof}

We will also use the following simple claim:
\begin{claim}\label{claim:yuv}
   If $G$ is a {linear} subspace of $\zo^n$ and $g_f$ is not anti-monotone, then 
$$\Prx_{\substack{\bu \sim F \\ \by \preceq x}}\big[f(\by \oplus \bu )=0\big] \geq 1/4,\quad\text{for any  $x \in {G} \setminus \mathcal{X}$.}$$
\end{claim}
\begin{proof}
    Note that $\LT(x)$ is a {linear} subspace, and since $x \not \in \mathcal{X}$ we have that $\LT(x) \not \subseteq G$. By \Cref{claim:PRS-15}: %
     $$ \Prx_{\by \preceq x}\big[\by \in G\big] \leq 1/2.$$
     Thus, if we sample $\by \preceq x$ uniformly at random, then with probability at least $1/2$ we have  $\by\notin G$ and $g_f(\by)=0$ which means $\Prx_{\bu \sim F}[f(\by \oplus \bu )=0] \geq 1/2$.
This finishes the proof of the claim.
\end{proof}

\begin{lemma}\label{lem: f cap g is large}
   Assume that $G$ is a linear subspace and $g_f$ is not an anti-monotone conjunction, and $\reldist(f,g_f) \leq 0.1$. Then we have$$\left|\left(G \setminus \mathcal{X}\right) \cap F\right| \geq 0.1N.$$
\end{lemma}
\begin{proof}
Since $g_f$ is not an anti-monotone conjunction, by \Cref{lemma:PRS02} we have $\left|G \setminus \mathcal{X}\right| \geq \left|G\right|/2$. But since $\reldist(f,g_f) \leq 0.1 $, we have $|G \setminus F|\le 0.1N$, and we also must have $|G| \geq 0.9N$. Hence  
$$\left|\left(G \setminus \mathcal{X}\right) \cap F\right| \geq 0.9N/2 - 0.1N \geq 0.1N.  $$ 
This finishes the proof of the lemma.
\end{proof}

We can now prove \Cref{lem: line 7-11}:

\begin{proofof}{\Cref{lem: line 7-11}}
    Assume that (i) $\reldist(f,f') > \eps$  for every anti-monotone conjunction $f'$; (ii) $\reldist(f,g_f) \leq \epsilon/10$; %
    and (iii) $G$ is a linear subspace. Then $G$ cannot be an anti-monotone linear subspace, as otherwise by \Cref{lem: monotone conj characterization} $g_f$ would be an anti-monotone conjunction and violate (i) and (ii).
    Hence we have $|(G \setminus \mathcal{X} ) \cap F | \geq N/10$ by \Cref{lem: f cap g is large}; this in turn means that for $\bx \sim F$, with probability at least $1/10$ we have $\bx \in G \setminus \mathcal{X}$. 

    By \Cref{claim:yuv}, for any $x \in G \setminus \mathcal{X}$ we have 
$$\Prx_{\substack{ \bu \sim F \\ \by \preceq x}}\big[f(\by \oplus \bu )=0\big] \geq 1/4.$$
Thus by \Cref{lem: f cap g is large} we have that
$$\Prx_{\substack{\bx,\bu \sim F \\ \by \preceq \bx  }}\big[f(\by \oplus \bu )=0\big] \geq 1/40.$$
Hence, each iteration of lines $4$--$8$ has at least a  $1/40$ probability to reject $f$. So the probability we fail to reject $f$ during lines $4$--$8$ is $(1-1/40)^{c_2}$, which is
    at most $0.1$ for a suitable choice of $c_2$.
\end{proofof}

\subsubsection{Putting the pieces together}
We have all the ingredients we need to establish the soundness of \Cref{algo: mono conjunction tester}: 
\begin{theorem}\label{thm: f isn't mono conj}
    Suppose that $f: \zo^n \to \zo$ satisfies $\reldist(f,f') >\eps$ for every anti-monotone conjunction $f'$. %
    Then \Cref{algo: mono conjunction tester} rejects $f$ with probability at least 0.9.
\end{theorem}
\begin{proof}
We prove the contrapositive, i.e.~that if \Cref{algo: mono conjunction tester} accepts $f$ with probability at least $0.1$ then $f$ has $\reldist(f,f') \leq \eps$ for some anti-monotone conjunction $f'$.

By \Cref{lem:f is close to g conj}, if $f$ passes Phase 1 with probability at least $1/10$, then $\reldist(f,g_f) \leq \epsilon/10$. 
By \Cref{lem:g is affine}, if $f$ passes Phase 1  with probability at least $1/10$, then $G$ is a linear subspace. 
Finally, by \Cref{lem: line 7-11} if $f$ passes Phase 2  with probability at least 0.1, then we must have $\reldist(f,f') \leq \eps$ for some anti-monotone conjunction $f'.$
\end{proof}
We can now prove \Cref{thm:conjunction} about \hyperlink{Algorithm2}{\sc Conj-Test}:
\begin{proofof}{\Cref{thm:conjunction}}
    Recall that \hyperlink{Algorithm2}{\sc Conj-Test} first samples $\by \sim F$ and then runs \Cref{algo: mono conjunction tester} on~$f_{\by}$. 
    When $f$ is a conjunction, $f_{\by}$ is an anti-monotone conjunction and thus, by \Cref{thm: f is mono conj} \hyperlink{Algorithm2}{\sc Conj-Test}   accepts with probability $1$.
    When  $f$~is far from any conjunction, $f_{\by}$ must be far from any anti-monotone conjunction. So by \Cref{thm: f isn't mono conj}, \hyperlink{Algorithm2}{\sc Conj-Test} rejects with probability at least $0.9$. 
\end{proofof}
\subsection{Robustness of  {\sc Conj-Test} to faulty oracles} 

\hyperlink{Algorithm2}{\sc Conj-Test} will play a crucial role in the relative-error
  tester for decision lists in the next section.
As will becomes clear there, \hyperlink{Algorithm2}{\sc Conj-Test} will be run
  on a function $f$ to test whether it is relative-error close to a conjunction when it can only
  access a ``faulty''
  version of the $\SAMP(f)$ oracle.
For later reference we record the following two explicit statements about
  the performance of \hyperlink{Algorithm2}{\sc Conj-Test};
  we show that as long as the faulty version of $\SAMP(f)$ satisfies some mild conditions,
  \hyperlink{Algorithm2}{\sc Conj-Test} still outputs the correct answer with high probability.
In both statements, we write $\calD$ to denote the distribution over $\{0,1\}^n$
  underlying the faulty sampling oracle $\SAMP^*(f)$, which is not necessarily
  uniform over $F=f^{-1}(1)$ (and {in \Cref{thm: reject weaker assumption}} need not even necessarily be supported over $F$).

\begin{theorem}\label{thm: accept weaker assumption}
    Let $f: \zo^n \to \zo$ be a conjunction.
    Assume that the oracle $\SAMP^\ast(f)$ always returns some string in $F$ (i.e., the  distribution $\calD$ is supported over $F$).
Then \hyperlink{Algorithm2}{\sc Conj-Test}, running on $\MQ(f)$ and $\SAMP^*(f)$,
accepts $f$ with probability 1.
\end{theorem}
\begin{proof}
When run on an anti-monotone conjunction $f$,  
\Cref{algo: mono conjunction tester} will never reject as long as all its samples are from $F$; therefore \Cref{algo: mono conjunction tester} will accept $f$ with probabilty $1$. 
Since \hyperlink{Algorithm2}{\sc Conj-Test} works by drawing a single sample $\by \sim \SAMP^*(f)$, which by assumption is in $F$, and running \Cref{algo: mono conjunction tester} on $f_{\by}$, iven that $\by \in F$, the function $f_{\by}$ is guaranteed to be an anti-monotone conjunction when $f$ is a conjunction, and \hyperlink{Algorithm2}{\sc Conj-Test} is guaranteed to accept.
\end{proof}

\begin{theorem}\label{thm: reject weaker assumption}
    Suppose that $f: \zo^n \to \zo$ satisfies $\reldist(f,f') > \epsilon$ for any conjunction $f'$. Let %
    $\SAMP^\ast(f)$ be a sampling oracle for $f$ such that the underlying distribution 
    $\calD$ satisfies
\[\Prx_{\bx \sim {\cal D}}\big[\bx = z\big] \geq \frac{1}{20|F|},\quad\text{for any $z\in F$}.\]
    Then \hyperlink{Algorithm2}{\sc Conj-Test}, running on $\MQ(f)$ and $\SAMP^*(f)$,
    rejects $f$ with probability at least $0.9$.
\end{theorem}
\begin{proof}
The proof of soundness
  (\Cref{thm: f isn't mono conj}) of \Cref{algo: mono conjunction tester} shows that, given any $f$ that satisfies
  $\reldist(f,f')>\eps$ for any anti-monotone conjunction $f$,
  either each round of Phase 1 rejects $f$ with probability $\Omega(1/\eps)$,
  or each round of Phase 2 rejects $f$ with probability $\Omega(1)$,
  under $\SAMP(f)$.
It is easy to show that under our assumption of $\SAMP^*$,
  both probabilities basically go down by at most a factor of $(1/20)^2$.
So by making $c_1$ and $c_2$ sufficiently large, we have that 
  \Cref{algo: mono conjunction tester} rejects $f$ with probability at least $0.9$
  even when it is given access to $\SAMP^*(f)$ instead of $\SAMP(f)$.

Now consider running \hyperlink{Algorithm2}{\sc Conj-Test} on a function $f$
  with $\reldist(f,f')>\eps$ for any conjunction $f'$, with access to a faulty oracle
  $\SAMP^*(f)$.    
\hyperlink{Algorithm2}{\sc Conj-Test} works by first drawing a single sample $\by \sim \SAMP^*(f)$ and running \Cref{algo: mono conjunction tester} on $f_{{\by}}$. Since $f$ is far from any conjunction, we must have $\reldist(f_{{\by}},f')>\epsilon$ for any anti-monotone conjunction $f'$ {(no matter whether $\by\in F$ or not)}. Hence, \hyperlink{Algorithm2}{\sc Conj-Test} is guaranteed to reject $f $ with probability at least $0.9$.
\end{proof}

\section{A relative-error testing algorithm for decision lists} \label{sec:DL}

Let's first recall how   
 a decision list is represented and some basic facts about decision lists.

\subsection{Basics about decision lists} \label{sec:basics-DL}

A decision list $f:\{0,1\}^n\rightarrow \{0,1\}$ can 
  be represented as a list of rules:%
$$
(x_{i_1}, b_1, v_1), \cdots, (x_{i_k}, b_k, v_k), v_{\textsf{default}},
$$
where $i_1,\ldots,i_k\in [n]$ are indices of variables
  and $b_1,\ldots,b_k,v_1,\ldots,v_k,v_{\textsf{default}}\in \{0,1\}$.
For each $x\in \{0,1\}^n$, $f(x)$ is set to be $v_j$ with the
  smallest $j$ such that $x_{i_j}=b_j$, and is set to be
  $v_{\textsf{default}}$ if $x_{i_j}\ne b_j$ for all $j\in [k]$.
Without loss of generality, we will make the following assumptions:
\begin{flushleft}\begin{itemize}
    \item \emph{Each variable appears at most once in the list of rules (i.e., no two $i_j$'s are the same).} Indeed if two rules of the form $(x_{i_j}, b_j,v_j)$ and $(x_{i_{\ell}}, b_{\ell}, v_{\ell})$ with $j<\ell$ and $i_j=i_{\ell}$ are present,  
    then either $b_j=b_{\ell}$, in which case 
    the latter rule can be trivially deleted,
    or $b_{\ell}=1-b_j$, in which case the latter rule and all rules after it can be replaced by $v_{\default}=v_{\ell}$, without changing the 
    function $f$ represented. 
    \item \emph{There exists at least one rule $(x_{i_j},b_j,v_j)$ with $v_j=1$.}
    This is because if $v_1=\cdots=v_k=0$, then either (1) $v_{\textsf{default}}=0$, in which case $f$ is the all-$0$ function; (2) $v_{\textsf{default}}=1$ and at least one variable $x_i$ does not appear in the list, in which case we can insert $(x_i,1,1)$ before $v_{\textsf{default}}$, which does not change the function and makes sure at least one $v_j$ is $1$; or (3) $v_{\textsf{default}}=1$ but all variables appear in the list, in which case we have $|f^{-1}(1)|=1$. Functions $f$ with $|f^{-1}(1)|\le 1$ are easy to recognize with the sampling oracle.
\end{itemize}\end{flushleft}
In the rest of the section, whenever we talk about 
  decision lists, the function $f$ is assumed to satisfy $|f^{-1}(1)|>1$ so it can be represented by a list of rules that satisfies the two conditions above.

Let $f:\{0,1\}^n\rightarrow \{0,1\}$ be a decision list
  represented by $(x_{i_1}, b_1,v_1),\cdots (x_{i_k}, b_k,v_k), v_{\textsf{default}}$.
The following decomposition of the list plays a crucial role both in the algorithm and its analysis:
\begin{flushleft}\begin{itemize}
    \item Let%
    $p\in [k]$ be the smallest index with $v_p = 1$. We will call $p$ the \textbf{pivot}. 
    \item The prefix before $(x_{i_p},b_p,v_{p})$, i.e., $(x_{i_1}, b_1, 0), \ldots, (x_{i_{p-1}}, b_{p-1}, 0)$ will be called the \textbf{head}. The key insight is that we can view the head as a conjunction $C:\{0,1\}^n\rightarrow \{0,1\}$, where $$C(x)= \big(x_{i_1}=1-b_1\big) \land \cdots \land \big(x_{i_{p-1}} = 1-b_{p-1}\big).$$ 
    And we write ${\cal H}_C:=\{{i_1}, \ldots {i_{p-1}}\}$ to denote the set of $p-1$
     variables in the conjunction.
    \item The part after $(x_{i_p},b_p,v_p)$, i.e., $(x_{i_{p+1}}, b_{p+1}, v_{p+1}),\cdots,(x_{i_k},b_k,v_k), v_{\default}$ will be called the $\textbf{tail}$.
Let $L:\{0,1\}^n\rightarrow \{0,1\}$ denote
  the decision list represented by the pivot and tail:
 $$(x_{i_p}, b_p,v_p=1),(x_{i_{p+1}}, b_{p+1},v_{p+1}), \cdots,(x_{i_k},b_k,v_k), v_{\default}.$$
Then it is easy to verify that $f(x)=C(x)\land L(x)$ for all $x\in \{0,1\}^n$.
\end{itemize}\end{flushleft}

We record the following simple fact  
  about the decomposition $f=C\land L$:

\begin{observation}\label{ob:simple1}
All $z\in f^{-1}(1)$ satisfy $C(z)=1$, i.e.,
  $z_{i_j}=1-b_j$ for all $i_j\in \calH_C$.
\end{observation}

We will need the following lemma about the decomposition $f=C\land L$ of a decision list, which says that coordinates in the tail are fairly unbiased. The intuition is that at least half the points in $f^{-1}(1)$ are accepted by the rule $(x_{i_p} , b_{p}, 1)$, and for these points the coordinates $x_{i_j}$ with $j>p$, do not matter, and thus are distributed uniformly. 
\begin{lemma}\label{lem: low bias tail}
    Let $f$ be a decision list represented
    by $(x_{i_1},b_1,v_1),\cdots,(x_{i_k},b_k,v_k),v_{\textsf{default}}$, and let $p$ be its pivot. 
    Then we have
\begin{equation*}\Prx_{\bz \sim f^{-1}(1)}\big[\bz_{i_p}=b_p\big] \geq 1/2\quad\text{and}\quad
    \Prx_{\bz \sim f^{-1}(1)}\big[\bz_{i_j}=1\big] \in [1/4, 3/4],\quad\text{for all $j:p<j\le k$.}
\end{equation*}
\end{lemma}
\begin{proof}
Using the decomposition $f=C\land L$, we have
\begin{equation}\label{eq:hehe2}
\big\{z\in \{0,1\}^n : C(z)=1\ \text{and}\  z_{i_p}=b_p \big\} \subseteq f^{-1}(1) \subseteq \big\{z\in \{0,1\}^n : C(z)=1\big\}.
\end{equation}
Given that  the LHS is exactly half of the size of the set on the RHS,
the first part follows.%

Furthermore, consider any $i_j$ with $p<j\le k$. Exactly half of the strings $z$ in the set on the LHS
  of \Cref{eq:hehe2} have $z_{i_j}=1$ and exactly half of them have $z_{i_j}=0$.
Thus, we have for any $b\in \{0,1\}$,
    \begin{align*}
        \Prx_{\bz \sim f^{-1}(1)}\big[\bz_{i_j}=b\big] &\geq \frac{|\{z: \text{$C(z)=1,z_{i_p}=b_p$ and 
        $z_{i_j}=b$}\}|}{|f^{-1}(1)|} 
         = \frac{|\{z: \text{$C(z)=1$}\}|}{4|f^{-1}(1)|} \ge \frac{1}{4},
    \end{align*}
where the last inequality used the second inclusion of 
  \Cref{eq:hehe2}. This finishes the proof.
\end{proof}

\subsection{The Algorithm}
\label{sec:DL-algorithm}

Our main algorithm for testing decision lists in relative distance is presented in \Cref{algo: Decision list tester},
  which uses two subroutines described in \Cref{algo:gamma_query} and \Cref{algo:gamma_sampler}, respectively.
Before diving into~the analysis of \Cref{algo: Decision list tester}, we start with an overview of the algorithm in this subsection, to give some intuition behind it about why it accepts functions that are decision lists and rejects functions that are far from decision lists in relative  distance with high probability.

We will start with the case when $f:\{0,1\}^n\rightarrow \{0,1\}$ is a decision list, and explain how the algorithm 
  extracts information about the underlying decomposition
  $C\land L$ of $f$ so that $f$ passes all the tests in
  \Cref{algo: Decision list tester} with high probability.
Towards the end of this subsection, we will switch to 
  the case when $f$ is far from any decision list in relative distance, and discuss why passing all the tests 
  in \Cref{algo: Decision list tester} gives evidence that 
  $f$ is close in relative distance to some decomposition
  $C\land L$ of a decision list, which should not happen given the assumption that $f$ is far from every decision list in relative distance.
The formal proof is given in \Cref{sec:queryandnonadaptivity,sec:completeness,sec:soundness}.

  \medskip
  \noindent {\bf The case when $f$ is a decision list.}
Let $f$ be a decision list represented by 
  $(x_{i_1},b_{1},v_1),\ldots,$ $(x_{i_k},b_k,v_k),v_{\textsf{default}}$.
Let $p$ be the pivot and $f=C\land L$ be the decomposition,
  where $C$ is a conjunction over variables in ${\cal H}_C:=\{i_1,\ldots,i_{p-1}\}$ and 
  $L$ is a decision list formed by the pivot and the tail of the list.

The first step of \Cref{algo: Decision list tester} draws samples from $\SAMP(f)$ to partition variables into $\bU\cup \bR$:\medskip

\begin{algorithm}[H]%
    \begin{algorithmic}[1] \vspace{0.2cm}
    \setcounter{ALG@line}{1}
            \State Draw $O(1/\eps)$ samples $\bz \sim \SAMP(f)$, and call the set of samples $\bS$. Let 
    $$\hspace{-1cm}\bU=\big\{i\in [n]  : \exists \, b\in \{0,1\} \text{ s.t. } \bz_i = b \ \text{for all $\bz \in \bS$} \big\} \quad\text{and}\quad
    \bR=[n] \setminus \bU$$ 
 and $\bu \in \{0,1\}^{\bU}$ be the unique string such that $\bu_i=\bz_i$ for all $i\in \bU$ and $\bz \in \bS$.\vspace{0.15cm}  
    \end{algorithmic}
\end{algorithm}\medskip

\noindent So $\bU$ is the set of \emph{unanimous coordinates} of 
  samples in $\bS$ and thus, we always have ${\cal H}_C\subseteq \bU$ and $\bu_{i_j}=1-b_j$ for all $i_j\in {\cal H}_C$
  by \Cref{ob:simple1}. 
Ideally, we would hope for $\bR$ to contain all variables in the tail, i.e., $i_{p+1},\ldots,i_k$; however, this won't be the case in general, unfortunately, since $k$ can be much bigger than $p$ (e.g., $k-p$ could be polynomial in $n$). %
But by \Cref{lem: low bias tail} and given that $\bS$ contains $O(1/\eps)$ samples from $\SAMP(f)$, it is easy to show that with high probability, the variables occurring in the beginning part of the tail %
(excluding those variables that are hidden too deep in the tail) all lie in $\bR$ with high probability, and it turns out that this is all we will need for the proof of completeness %
to go through.

Now, the pivot variable $x_{i_p}$ can go either way, since 
  it may be highly biased towards $b_{p}$.
However, if it is too biased towards $b_{p}$ so that
  none of the samples $\bz$ in $\bS$ have $\bz_p=1-b_p$, then one can argue that $f$ must be very close
  to a conjunction over $\{0,1\}^n$ in relative distance.
Because of this, Step $0$ of \Cref{algo: Decision list tester}
  will accept $f$ with high probability.

Proceeding to Step 2, assume without loss of generality that
  \Cref{algo: Decision list tester} has found $\bU$ and $\bR$ such that ${\cal H}_C\subseteq \bU$ and $\bR$ contains variables in the beginning part of $L$, including $i_p$, and the string $\bu\in \{0,1\}^{\bU}$ satisfies $\bu_{i_j}=1-b_j$ for all $i_j\in {\cal H}_C$.
This implies the following properties:
\begin{flushleft}\begin{enumerate}
\item The function $f{\upharpoonleft_{\bu}}$ over $\{0,1\}^{\bR}$  is a ``dense'' decision list, where 
  ``dense'' means the number\\ of its satisfying assignments 
  is at least half of $2^{|\bR|}$, and it is a decision list
  because the restriction of any decision list remains a decision list.
\item For any $\alpha\in \{0,1\}^{\bU}$, $f{\upharpoonleft_{\alpha}}$ is either the all-$0$ function
  (because $\alpha$ does not satisfy the conjunction $C$), or very close to $f{\upharpoonleft_{\bu}}$  (because $\bR$ contains the beginning part of the tail) and thus, both $f{\upharpoonleft_{\bu}}$
  and $f{\upharpoonleft_{\alpha}}$ are close to $L$ (when viewed as a decision list over $\bR$) under the uniform distribution.
\end{enumerate}\end{flushleft}

Given these properties of $\bU,\bR$ and $\bu$, it is easy to see that $f$ passes both tests in Step 2 with high probability:
\medskip

\begin{algorithm}[H]%
    \begin{algorithmic}[1]\vspace{0.2cm}
    \setcounter{ALG@line}{2}
        \State Draw  $O(1)$ points $\bw\sim \{0,1\}^{\bR}$;
    halt and reject 
    $f$
if less than $1/4$ of the sampled $\bw$'s have  $f(\bu\circ \bw)=1$. %
    \State Run the algorithm of \Cref{thm:Bshouty20} on $f{\upharpoonleft_{\bu}}$ with $\epsilon/100$;
     halt and reject $f$ if it rejects.\vspace{0.15cm}
    \end{algorithmic}
\end{algorithm}\medskip

\noindent since line 3 checks whether $f{\upharpoonleft_{\bu}}$ is ``dense,'' and line 4 checks if $f{\upharpoonleft_{\bu}}$ is close to a decision list under the uniform distribution. The function
$f$ would also pass  the tests in Step 3 
 with high probability: \medskip

\begin{algorithm}[H]%
    \begin{algorithmic}[1]\vspace{0.2cm}
    \setcounter{ALG@line}{4}
        \State Repeat $O(1/\eps)$ times:
     \State \hskip2em Draw $\bz \sim \SAMP(f)$ and $\bw \sim \{0,1\}^{\bR}$; if $f(\bz_{\bU} \circ \bw) \neq f(\bu \circ \bw)$, halt and reject\vspace{0.15cm} $f$.
    \end{algorithmic}
\end{algorithm}\medskip
\noindent as for any $z\in f^{-1}(1)$, $z_{\bU}$ satisfies $C$ and thus,
$f{\upharpoonleft_{z_{\bU}}}$ is close to $f{\upharpoonleft_{\bu}}$ under the uniform distribution.

For step $4$, let $\bGamma:\{0,1\}^{\bU}\rightarrow \{0,1\}$ be
  defined as follows:
$$\bGamma(\alpha)=\begin{cases}
         1 & \text{if $\Pr_{\bw \sim \{0,1\}^{\bR} }\big[f(\alpha \circ \bw)=1\big] \ge 1/16$} \\
         0 &\text{otherwise}
    \end{cases}$$
Now, observe that for any $\alpha\in \{0,1\}^{\bU}$,  
  either $\alpha$ does not satisfy $C$, in which case
$f{\upharpoonleft_{\alpha}}$ is the all-$0$ function,
  or $\alpha$ satisfies $C$, in which case 
$f{\upharpoonleft_{\alpha}}$ is close to 
  $f{\upharpoonleft_{\bu}}$ and thus
  $\Gamma(\alpha)=1$ (because $f{\upharpoonleft_{\bu}}$ is ``dense'').
As~a result, $\bGamma$ is indeed the same as $C$ and thus it is a conjunction over $\{0,1\}^{\bU}$. From this, it may seem straightforward to conclude that $\Gamma$
  always passes the test in Step 4 below:\medskip 
  
\begin{algorithm}[H]%
    \begin{algorithmic}[1]\vspace{0.2cm}
    \setcounter{ALG@line}{6}
\State Run \hyperlink{Algorithm2}{\sc Conj-Test}
    with $\SAMP(\bGamma) \leftarrow $(\Cref{algo:gamma_sampler}), $\MQ(\bGamma) \leftarrow $(\Cref{algo:gamma_query})  to test if  
    $\bGamma$ \Statex  is $(\epsilon/100)$-close to a conjunction in relative distance; halt and return the same answer. \vspace{0.15cm} 
    \end{algorithmic}
\end{algorithm}\medskip

\noindent except that we don't actually have direct access to 
  sampling and query oracles of $\bGamma$.
They are simulated by \Cref{algo:gamma_sampler} and \Cref{algo:gamma_query}, respectively.
We show in the proof that (1) \Cref{algo:gamma_sampler} (for the sampling oracle of $\bGamma$) only returns samples in 
  $\bGamma^{-1}(1)$; and (2) with high probability,
  \Cref{algo:gamma_query} returns the correct answer to all $O(1/\eps)$ membership queries made by \hyperlink{Algorithm2}{\sc Conj-Test}.
From this we can apply \Cref{thm: accept weaker assumption} %
to conclude that \hyperlink{Algorithm2}{\sc Conj-Test} accepts with high probability.

\medskip
\noindent {\bf The case when $f$ is far from every decision list.}
Let's now switch to the case when $f$ is $\eps$-far from
  any decision list in relative distance.
Because conjunctions are special cases of decision lists,
  Step 0 does not accept $f$ with high probability and 
  \Cref{algo: Decision list tester}~goes through Step 1 to obtain $\bU,\bR$ and $\bu\in \{0,1\}^{\bU}$.
Let's consider what conditions $f$ needs to satisfy 
  in order for it to pass tests in Steps 2 and 3 with high probability:
\begin{flushleft}
\begin{enumerate}
\item To pass Step 2,   $f{\upharpoonleft_{\bu}}$  over $\{0,1\}^{\bR}$ needs to be ``dense:'' $|f{\upharpoonleft_{\bu}}^{(-1)}(1)|\ge 2^{|\bR|}/8$, and 
  must also be $(\eps/100)$-close to a decision list under the uniform distribution.
\item To understand the necessary condition posed on $f$ to 
  pass Step 3, let's introduce the following function $\bg:\{0,1\}^n\rightarrow \{0,1\}$, which is closely related to $\bGamma$ defined earlier and\\ will play a crucial role in the analysis:
$$
\bg(x)= \bGamma(x_{\bU}) \wedge f(\bu\circ x_{\bR}). 
$$
We show that to pass Step 3 with high probability,
  it must be the case that
  $\reldist(f,\bg)$ is small.
This means that for each $\alpha\in \{0,1\}^{\bU}$, one can replace $f{\upharpoonleft_{\alpha}}$ by the 
  all-$0$ function if $\bGamma(\alpha)=0$ (which means that $f{\upharpoonleft_{\alpha}}$ is not dense enough),
  and replace $f{\upharpoonleft_{\alpha}}$ by 
  $f{\upharpoonleft_{\bu}}$ otherwise, and this will only change a small number of bits of $f$ (relative to $|f^{-1}(1)|$). 
  Note that $\bg$ is getting closer to the decomposition of a decision list 
  since we know $f{\upharpoonleft_{\bu}}$ is close to a decision list over $\{0,1\}^{\bR}$.
\end{enumerate}
\end{flushleft}

Assuming all the necessary conditions summarized above for $f$ to pass Steps 2 and 3
and using the assumption that $f$ is far from decision lists in relative distance,
  we show that $\bGamma$ must be far from conjunctions in relative distance.
Finally, given the latter, it can be shown that \hyperlink{Algorithm2}{\sc Conj-Test} in Step 4 rejects with high probability, again by analyzing performance guarantees of \Cref{algo:gamma_sampler} and \Cref{algo:gamma_query} as simulators of sampling and query oracles of $\bGamma$, respectively, so that \Cref{thm: reject weaker assumption}  for \hyperlink{Algorithm2}{\sc Conj-Test} %
can be applied. 

\begin{algorithm}
\caption{Relative-error decision list tester. 
Here $c_1$ and $c_2$ are  absolute constants\\ chosen at the beginning of \Cref{sec:queryandnonadaptivity}.} \label{algo: Decision list tester}
\vspace{0.15cm} \textbf{Input:}
 Oracle access to $\MQ(f),$ $\SAMP(f)$ of $f:\{0,1\}^n\rightarrow \{0,1\}$ and a parameter $\eps$. \\
\textbf{Output:} ``Reject'' or ``Accept.''

 \begin{tikzpicture}
\draw [thick,dash dot] (0,1) -- (15.3,1);
\end{tikzpicture}
\begin{algorithmic}[1]\vspace{-0.15cm}
    \Algphase{Step 0: Special case when $f$ is close to a conjunction\vspace{-0.15cm} }
    \State Run \hyperlink{Algorithm2}{\sc Conj-Test}
    on $f$ with $\epsilon$;
    halt and accept $f$ if it
    accepts.\vspace{-0.15cm}
    \Algphase{Step 1: Draw samples to obtain $\bU$ and $\bR$\vspace{-0.15cm}}
    \State Draw ${c_1/\eps}$ samples $\bz \sim \SAMP(f)$, and call the set of samples $\bS$. Let 
    $$\hspace{-1cm}\bU=\big\{i\in [n]  : \exists \, b\in \{0,1\} \text{ s.t. } \bz_i = b \ \text{for all $\bz \in \bS$} \big\} \quad\text{and}\quad
    \bR=[n] \setminus \bU$$ 
 and $\bu \in \{0,1\}^{\bU}$ be the unique string such that $\bu_i=\bz_i$ for all $i\in \bU$ and $\bz \in \bS$.  \vspace{-0.15cm}
    \Algphase{Step 2: Check that $f{\upharpoonleft_{\bu}}$ is close to a ``dense'' decision list \vspace{-0.15cm}}
    \State Draw  ${c_2}$ points $\bw\sim \{0,1\}^{\bR}$;
    halt and reject 
    $f$
if less than $1/4$ of the sampled $\bw$'s have  $f(\bu\circ \bw)=1$. %
    \State Run the algorithm of \Cref{thm:Bshouty20} on $f{\upharpoonleft_{\bu}}$ with $\epsilon/100$;
     halt and reject $f$ if it rejects.\vspace{-0.15cm}
     \Algphase{Step 3: Check that 
     $f{\upharpoonleft_{w}}$ is close to $f{\upharpoonleft_{\bu}}$ on ``most'' $w \in \zo^{\bU}$
     \vspace{-0.15cm}} %
     \State Repeat ${c_2/\eps}$ times:
     \State \hskip2em Draw $\bz \sim \SAMP(f)$ and $\bw \sim \{0,1\}^{\bR}$; if $f(\bz_{\bU} \circ \bw) \neq f(\bu \circ \bw)$, halt and reject\vspace{-0.15cm} $f$. 

    \Algphase{Step 4: Is $\bGamma$ a conjunction?\vspace{-0.15cm}}
    \Statex Let 
    $\bGamma : \{0,1\}^{\bU} \to \{0,1\}$ be defined as follows:
$$\bGamma(\alpha)=\begin{cases}
         1& \text{if $\Pr_{\bw \sim \{0,1\}^{\bR} }\big[f(\alpha \circ \bw) =1\big] \ge  1/16$} \\ 
         0 &\text{otherwise.}
    \end{cases}$$

    \State Run \hyperlink{Algorithm2}{\sc Conj-Test}
    with $\SAMP(\bGamma) \leftarrow $(\Cref{algo:gamma_sampler}), $\MQ(\bGamma) \leftarrow $(\Cref{algo:gamma_query})  to test if  
    $\bGamma$ \Statex  is $(\epsilon/100)$-close to a conjunction in relative distance; halt and return the same answer. \vspace{0.15cm}
\end{algorithmic}
\end{algorithm}

\begin{algorithm}\caption{Simulator for query access to $\bGamma$.}\label{algo:gamma_query} \textbf{Input: }Oracle access to $\MQ(f)$ and a query $\alpha \in \{0,1\}^{\bU}$. \\
 \begin{tikzpicture}
\draw [thick,dash dot] (0,1) -- (15.5,1);
\end{tikzpicture}
    \begin{algorithmic}[1]
        \State Draw $\bw \in \{0,1\}^{\bR}$ for {$c_2\log(1/\epsilon)$} many times and return $0$ if $f(\alpha\circ \bw)=0$ for all $\bw$.%
        \State Halt and reject $f$ if $f(\alpha \circ \bw) \neq f(\bu \circ \bw)$ for some $\bw$; return $1$ otherwise. %
        
    \end{algorithmic}
\end{algorithm}

\begin{algorithm}\caption{Simulator for sample access to $\bGamma$.}\label{algo:gamma_sampler}
 \textbf{{Input: }}Oracle access to $\SAMP(f)$. \\
 \begin{tikzpicture}
\draw [thick,dash dot] (0,1) -- (15.5,1);
\end{tikzpicture}
    \begin{algorithmic}[1]
        \State Draw $\bz \sim \SAMP(f)$ and  
 return $\bz_{\bU}$.
    \end{algorithmic}
\end{algorithm}

\subsection{Query Complexity}\label{sec:queryandnonadaptivity}

Let $c_0$ be the absolute constant such that running 
  \hyperlink{Algorithm2}{\sc Conj-Test} with relative-error distance
  parameter $\eps/100$ uses no more than $c_0/\eps$ many
  samples and membership queries.
Then the two constants $c_1$ and $c_2$ are chosen to satisfy the 
  following conditions:
 1) $c_2$ is both sufficiently large and sufficiently larger
  than $c_0$;
 2) $c_1$ is chosen to be sufficiently larger than $c_2$ (and thus, $c_1$ is also sufficiently large and sufficiently larger than $c_0$).

The following lemma gives the easy efficiency analysis of \Cref{algo: Decision list tester}:
\begin{lemma}
    \Cref{algo: Decision list tester}  makes no more than $\tilde{O}(1/\epsilon)$ calls to $\SAMP(f)$ and $\MQ(f)$.
\end{lemma}
\begin{proof}
    Running 
    \hyperlink{Algorithm2}{\sc Conj-Test}
    in line $1$ takes $O(1/\epsilon)$ samples and non-adaptive queries. Line $2$ uses 
    ${O(1/\eps)}$
    samples.
    Line {$3$} uses $O(1)$ queries. Line {$4$} uses $\tilde{O}(1/\epsilon )$ queries by \Cref{thm:Bshouty20}. Line ${6}$ uses $O(1/\epsilon)$ samples and queries. 
    \hyperlink{Algorithm2}{\sc Conj-Test}
    on line ${7}$ runs
    \Cref{algo:gamma_query} and \Cref{algo:gamma_sampler}~each $O(1/\eps)$ times, which in total uses $\tilde{O}(1/\eps)$
    samples and queries to $\SAMP(f)$ and $\MQ(f)$.
\end{proof}
\begin{remark} \label{remark:DL-nonadaptive}
An alternative version of \Cref{algo: Decision list tester} is obtained simply by replacing the algorithm of \Cref{thm:Bshouty20}
with the algorithm of \Cref{thm:DLM+:07}.  As discussed in \Cref{sec:results-techniques}, this yields a version of \Cref{thm:DL} which is nonadaptive and uses $\tilde{O}(1/\eps)$ samples and makes $\tilde{O}(1/\eps^2)$ calls to $\MQ(f).$
\end{remark}

\subsection{Completeness: 
 Correctness of \Cref{algo: Decision list tester} when $f$ is a decision list}
 \label{sec:completeness}

Throughout this section, we assume $f$ to be a decision list with 
  pivot $p$ and decomposition~$C\land L$,
  where
$C:\{0,1\}^n\rightarrow \{0,1\}$ is a conjunction over variables in ${\cal H}_C=\{i_1,\ldots,i_{p-1}\}$ and 
$L:\{0,1\}^n\rightarrow \{0,1\}$ is a decision list over other variables,
  represented by the following list of rules: %
 $$(x_{i_p}, b_{p},v_p), \ldots, (x_{i_k}, b_k, v_k), v_{\default}$$
with $v_p=1$ given that $p$ is the pivot.
  Our goal is to show the following:
\begin{theorem}\label{thm: when f is a DL}
If $f$ is a decision list, then \Cref{algo: Decision list tester} accepts $f$ with probability at least $0.7$.
\end{theorem}

To prove \Cref{thm: when f is a DL}, we will  consider two cases (recall that $c_0$ is the constant such that $c_0/\eps$~is the number of 
    queries and samples asked by \hyperlink{Algorithm2}{\sc Conj-Test} with relative error parameter $\eps/100$):
$$
\text{Case $1$:}\quad
\Prx_{\bz \sim f^{-1}(1)}\big[\bz_{i_p} \neq b_p\big] \leq \frac{\epsilon}{10c_0};
\qquad\text{ Case $2$:}\quad 
 \Prx_{\bz \sim f^{-1}(1)}\big[\bz_{i_p} \neq b_p\big] > \frac{\epsilon}{10c_0}.
$$
The key idea, is that in the first case, $f$ is so close to a conjunction that \hyperlink{Algorithm2}{\sc Conj-Test}
will accept in line 1 with high probability. In the second case, with high probability we will have $i_p \in \bU$. 

\subsubsection{Case 1}
We will first prove a lemma regarding our conjunction tester.

\begin{lemma}\label{lem:simple1}
    Let $h : \{0,1\}^n \to \{0,1\}$ and %
    let $D$ be an anti-monotone conjunction with $D^{-1}(1)\subseteq h^{-1}(1)$. Consider running \Cref{algo: mono conjunction tester} on $h$ using $\MQ(h)$ and $\SAMP(h)$. 
    If all samples \Cref{algo: mono conjunction tester} received lie in $D^{-1}(1)$, then it always accepts $h$.%
\end{lemma}
\begin{proof}
    The algorithm can only reject on lines 3 and 8.
    On lines $3$ assumption on samples we have $\bx, \by \in D^{-1}(1)$. Thus $\bx \oplus \by \in D^{-1}(1)$ so $h(\bx \oplus \by)=1$. 
    For line 11, we have $\bx,\bu \in D^{-1}(1)$, and since $\by \preceq \bx$ we have that $\by\in D^{-1}(1)$.
So again, $h(\by \oplus \bu)=1$.
\end{proof}

\begin{corollary}\label{lem:samples and query from conj}
    Let $h : \{0,1\}^n \to \{0,1\}$, and  $D$ be a  conjunction 
with $D^{-1}(1)\subseteq h^{-1}(1)$. %
    Consider running \hyperlink{Algorithm2}{\sc Conj-Test}
    on $h $ using $\MQ(h)$ and $\SAMP(h)$.
If all samples that \hyperlink{Algorithm2}{\sc Conj-Test} received lie in $D^{-1}(1)$, then 
    it always
    accepts $h$.
\end{corollary}
\begin{proof}
    Let $D'$ be the anti-monotone version of $D$, i.e.~$D'$ is obtained from $D$ by replacing every positive literal
      $x_i$ with the negated version $\overline{x_i}$. %
Recall that 
    \hyperlink{Algorithm2}{\sc Conj-Test}  first draws
 $\by \sim h^{-1}(1)$, which by assumption satisfies $\by\in D^{-1}(1)$, and then runs \Cref{algo: mono conjunction tester} on $h_{\by}$ (cf.~\Cref{remark:monotone}).
$D^{-1}(1)\subseteq h^{-1}(1)$ implies that $\smash{D'^{-1}(1)\subseteq h^{-1}_{{\by}}(1)}$.
Moreover,
since all samples lie in $D^{-1}(1)$, all  samples to \Cref{algo: mono conjunction tester}~must~lie~in $D'^{-1}(1)$. 
 It follows from \Cref{lem:simple1} that
    \hyperlink{Algorithm2}{\sc Conj-Test}
    accepts.
\end{proof}

We are ready to prove \Cref{thm: when f is a DL} for Case 1:

\begin{lemma}\label{lem: f is a DL case 1}
    Let $f=C\land L$ be a decision list as described at the start of \Cref{sec:completeness} with $$\Prx_{\bz \sim f^{-1}(1)}\big[\bz_{i_p} \neq b_p\big] \leq \frac{\epsilon}{10c_0}\,.$$ Then
    \Cref{algo: Decision list tester} accepts $f$ on line 1 with probability at least $0.9$. %
\end{lemma}
\begin{proof}
 Consider the conjunction $D=C \land (x_{i_p}=b_p)$. We have $D^{-1}(1) \subseteq f^{-1}(1)$.

It follows from \Cref{lem:samples and query from conj} that 
  \hyperlink{Algorithm2}{\sc Conj-Test} accepts $f$ on line 1
  if all the $c_0/\epsilon$ samples drawn from $f^{-1}(1)$ lie 
  in $D^{-1}(1)$.
{(Recall that $c_0/\eps$ is the number of samples and queries 
  asked by \hyperlink{Algorithm2}{\sc Conj-Test}  
  with parameter $\eps/100$. So $c_0/\eps$ is also an upper bound for the number of samples \hyperlink{Algorithm2}{\sc Conj-Test} draws in line 1  with  $\eps$.)}
Note that every $z\in f^{-1}(1)$ satisfies $z\in D^{-1}(1)$ as long as 
  $z_{i_p}=b_p$.
So all $c_0/\eps$ samples lie in $D^{-1}(1)$ with probability at least $0.9$ by a union bound.
\end{proof}

\subsubsection{Case 2}
\label{sec:case2}

In the following analysis we set $\ell$ to be the following parameter: 
$$
\ell:= \min\left(k -p,\hspace{0.06cm}3\log\left(\frac{c_1}{\epsilon}\right)\right)\le 3\log\left(\frac{c_1}{\epsilon}\right).%
$$
\begin{lemma}\label{lem: U and R are good}
    Let $f=C\land L$ be a decision list as described at the start of \Cref{sec:completeness}. Suppose that%
    $$\Prx_{\bz \sim f^{-1}(1)}\big[\bz_{i_p} \neq b_p\big] \geq \frac{\epsilon}{10c_0}\,.$$
    If the algorithm reaches line $3$, then $i_{p+j}\in \bR$ for all $j:0\le j\le \ell$ with probability at least $0.9 $.
\end{lemma}
\begin{proof}
The event follows if for all $i\in \{i_p,\ldots,i_{p+\ell}\}$ and $b\in \{0,1\}$,
  at least one $\bz\in\bS$ has $\bz_i=b$.

    On line $2$, we sample ${c_1/\epsilon}$ many $\bz \sim f^{-1}(1)$ with $c_1$ sufficiently larger than $c_0\ge 1$.
First we handle $i=i_p$: because we are in Case 2, the probability of $\bz_{i_p}=b_p$ for all $\bz \in \bS$ is at most
    $$\left(1-\frac{\epsilon}{10c_0}\right)^{{c_1/\eps}}.$$
    By \Cref{lem: low bias tail}, $\Pr_{\bz \sim f^{-1}(1)}[\bz_{i_p}=b_p] \geq 1/2$. Hence the probability of $\bz_{i_p}=1-b_p$ for all $\bz$ is at most 
     $\smash{(1/2)^{{c_1/\epsilon}}}$.
Setting $c_1$ to be sufficiently larger than $c_0$, we have $i_p\notin \bR$ with probability at most $0.05$.

For each $i_{p+j}$ with $j>0$, by \Cref{lem: low bias tail} we have for either $b \in \{0,1\}$, $$\Prx_{\bz\sim f^{-1}(1)}\big[ \text{$\bz_{i_{p+j}}=b$ for all $\bz\in \bS$}\big ] \leq \left(\frac{3}{4}\right)^{{c_1/\epsilon}}.$$
The probability of each of these $2\ell$ events can be made
  sufficiently smaller than $1/\ell$ (given that $\ell$ is only logarithmic in $c_1/\eps$) by taking $c_1$ sufficiently large.
The lemma follows from a union bound.
\end{proof}

We will also need the following simple observation about $\bU$:
\begin{lemma}\label{lem: H is good}
    If \Cref{algo: Decision list tester} reaches line 3, then
    $\mathcal{H}_C\subseteq \bU$ and $\bu_{i_j}=1-b_{j}$ for all $i_j\in {\cal H}_C$.
\end{lemma}
\begin{proof}
    Given that $f=C\land L$, by definition  we must have that
    every $z\in f^{-1}(1)$ satisfies $z_{i_j}=1-b_{j}$ for all $i_j\in {\cal H}_C$. 
So this must also be the case for points in $\bS$.
\end{proof}

For all lemmas in the rest of this section we will assume that the conclusions of \Cref{lem: U and R are good} and \Cref{lem: H is good} hold, i.e.:%

\begin{assumption}\label{assumption: U R are good}
\Cref{algo: Decision list tester} reaches line 3 with $\bU=U,\bR=R$ and $\bu=u$ satisfying%
    \begin{itemize}
        \item $\mathcal{H}_C \subseteq  U$ and $ u_{i_j}=1-b_{j}$ for all $i_j\in {\cal H}_C$.
        \item For any $j:p \leq j \leq  p+\ell$, $i_j \in  R$.
    \end{itemize}
Note that once $U,R$ and $u$ are fixed, 
  the function $\bGamma=\Gamma$ over $\{0,1\}^U$ is also fixed. 
\end{assumption}

The next lemma will play a key role in the rest of the proof; it says that under \Cref{assumption: U R are good}, one can ignore the rules of $L$ that are too deep in the list.
\begin{lemma}\label{thm: subcubes look the same}
Under \Cref{assumption: U R are good},
for any $\alpha\in \{0,1\}^{U}$ with $\alpha_i=1-b_i$ for all $i\in  \mathcal{H}_C$, we have 
\begin{align*}
    \Prx_{\bw \sim \{0,1\}^{R}}\big[f(\alpha \circ \bw) \neq f(u \circ \bw)\big] \leq %
    \left(\frac{\eps}{ c_1}\right)^3.
\end{align*}
\end{lemma}
\begin{proof}
Notice that both $u$ and $\alpha$ satisfy the conjunction given by $C$. Hence we have $f(u \circ w)=L(u \circ w)$ and $f(\alpha \circ w)=L(\alpha \circ w)$ for every $w\in\{0,1\}^R$. By \Cref{assumption: U R are good}, we have $i_p, \ldots, i_{p+\ell} \in R$. 
On the one hand, if $\ell <3\log (c_1/\eps)$, then we have $\ell=k-p$ and thus,  $L(\alpha\circ w)=L(u\circ w)$ for all $w\in \{0,1\}^{R}$ in which case the probability we are interested in is 
  trivially $0$.
On the other hand, if $\ell \geq 3\log (c_1/\eps)$ then if $w_{i_j}=b_{i_j}$ for some $j: p \leq j \leq p + \ell$ we have $L(\alpha \circ w)=L(u \circ w)$. Thus, 
    \begin{align*}
        \Prx_{\bw \sim \{0,1\}^{R}}\big[f(\alpha \circ \bw) \neq f(u \circ \bw)\big] &\leq \Prx_{\bw \sim \{0,1\}^{R}}\big[ \text{$\bw_{i_j} \neq b_{i_j}$ for all $j:p \leq j \leq p + \ell$}\big] = 2^{-(\ell+1)}. 
    \end{align*} 
Plugging in  $\ell=3\log (c_1/\eps)$   finishes the proof of the lemma in this case.
\end{proof}

Next we show that, under \Cref{assumption: U R are good}, \Cref{algo: Decision list tester} is unlikely to reject during lines 3--6:

\begin{lemma}\label{lem: f is DL 6-7}
Under \Cref{assumption: U R are good}, \Cref{algo: Decision list tester} rejects $f$ on lines 3--6 with probability at most $0.1$.
\end{lemma}
\begin{proof}
    Since $i_p \in R$ and $ u$ satisfies the conjunction in $C$, we have that $$\Prx_{\bw \sim \{0,1\}^R}\big[f(u\circ \bw)=1\big] \geq \Prx_{\bw \sim \{0,1\}^R}\big[{\bw_{i_p}=b_p} %
    \big] = 1/2.$$
    Hence, by a simple Chernoff bound and by making the constant $c_2$ sufficiently large, we have that the probability that less than $1/4$ of the points queried have $f(\bz)=1$ is at most $0.01$.

Next for line 4, note that by \Cref{assumption: U R are good}, we have that $f{\upharpoonleft_{u}}=(C\land L){\upharpoonleft_{u}}=L{\upharpoonleft_{u}}$. But $L{\upharpoonleft_{u}}$ is a decision list. %
Hence line $4$ will reject with probability at most $0.05$ by  \Cref{thm:Bshouty20}.

   Recall that line $6$ is repeated $c_2/\epsilon$ times.
   Fix any $z \in f^{-1}(1)$ as the point sampled in line ${6}$.  By \Cref{thm: subcubes look the same}, we have that 
    \begin{align*}
    \Prx_{\bw \sim \{0,1\}^{R}}\big[f(u \circ \bw) \neq f(z_{U} \circ \bw)\big] \leq \left(\frac{\epsilon}{ c_1}\right)^3.
\end{align*}
By a union bound over the $c_2/\eps$ rounds on line 11 and noting
  that $c_1$ is sufficiently larger than $c_2$,
  the algorithm rejects with probability at most $0.01$.
The lemma follows by a union bound.
\end{proof}

Finally it remains to examine line 7. 
Let $C^*:\{0,1\}^U\rightarrow \{0,1\}$ be the restriction of $C$, which~has domain $\{0,1\}^n$ but is just a conjunction over $\calH_C\subseteq U$, onto the domain $\{0,1\}^U$:
$$
C^*(\alpha)=\big(\alpha_{i_1}=1-b_1\big)\land \cdots \land \big(\alpha_{i_{p-1}}=1-b_{p-1}\big),
$$
{for any $\alpha\in \{0,1\}^U$.}
Our first goal will be to show that $\Gamma=C^*$. 

\begin{lemma}\label{thm: gamma=C}
Under \Cref{assumption: U R are good}, we have that  $\Gamma=C^*$ (so in particular, $\Gamma$ is a conjunction).
\end{lemma}
\begin{proof}
If $C^*(\alpha)=0$, then  
  $f(\alpha\circ w)=0$ for all $w\in \{0,1\}^R$ and by definition  $\Gamma(\alpha)=0$.

On the other hand, if $C^*(\alpha)=1$, it follows from
  \Cref{assumption: U R are good} and \Cref{thm: subcubes look the same} that 
$$
\Prx_{\bw\in \{0,1\}^R}
\big[f(\alpha\circ \bw)=1\big]\ge 
\Prx_{\bw\in \{0,1\}^R}
\big[f(u\circ \bw)=1\big]-\left(\frac{\eps}{c_1}\right)^3
\ge \frac{1}{2}-\left(\frac{\eps}{c_1}\right)^3\ge \frac{1}{16}
$$
when $c_1$ is sufficiently large and thus, $\Gamma(\alpha)=1$.
\end{proof}

By \Cref{thm: gamma=C}, if we were to run 
\hyperlink{Algorithm2}{\sc Conj-Test}
on $\Gamma$ with the correct membership and sampling oracles, we know it would always accept (recall that \hyperlink{Algorithm2}{\sc Conj-Test} has one-sided error). However,~we need to account for the fact we are only able to simulate access to oracles $\SAMP(\Gamma)$ and $\MQ(\Gamma)$ using \Cref{algo:gamma_sampler}  and \Cref{algo:gamma_query}, respectively.  

\begin{lemma}\label{lem: f is DL all queries are good.}
    Suppose that $f$ reaches line $8$ of \Cref{algo: Decision list tester} under \Cref{assumption: U R are good}. 
     Then  \hyperlink{Algorithm2}{\sc Conj-Test}  accepts with probability at least $0.9$.
\end{lemma}
\begin{proof}
By \Cref{thm: gamma=C}, all samples \hyperlink{Algorithm2}{\sc Conj-Test} receives using \Cref{algo:gamma_sampler} are from $\Gamma^{-1}(1)$ because they are from $\smash{{C^*}^{-1}(1)}$ and $\Gamma=C^*$.
We prove below that with probability at least ${9/10}$,
  all the {$c_0/\eps$}
  membership queries made by 
  \hyperlink{Algorithm2}{\sc Conj-Test} are correctly answered by \Cref{algo:gamma_query}.
When this happens, it follows directly from 
  \Cref{thm: accept weaker assumption} that
  \hyperlink{Algorithm2}{\sc Conj-Test} accepts and the lemma follows.
For this purpose, it suffices to show that given any $\alpha \in \{0,1\}^{U}$, the probability that \Cref{algo:gamma_query}   
  rejects or returns the wrong value for $\Gamma(\alpha)$ is at most $ {{\epsilon}/({10c_0})}$.
  The lemma then follows by a union bound over all ${c_0/\eps}$ %
  membership queries made by \hyperlink{Algorithm2}{\sc Conj-Test}.
  
 If $\Gamma(\alpha)=0$, then it follows from \Cref{thm: gamma=C}
 that   we have $f(\alpha\circ w)=0$ for all $w\in \{0,1\}^R$ and thus   \Cref{algo:gamma_query}  always returns $0$.
 
    If $\Gamma(\alpha)=1$, then $C^*(\alpha)=1$ and thus, the probability of $\bw\sim \{0,1\}^R$ having $f(\alpha\circ \bw)=1$ is at least $1/2$.
    So the probability of \Cref{algo:gamma_query} returning $0$ is at most $2^{-{ c_2\log(1/\epsilon)}} \leq {\epsilon/(20c_0)}$ given that $c_2$ is sufficiently larger than $c_0$.
    On the other hand,
   by \Cref{thm: subcubes look the same} we have 
    $$\Prx_{\bw \sim \{0,1\}^n}\big[f(\alpha \circ \bw) \neq f(u \circ \bw)\big] \leq
    \left(\frac{\eps}{ c_1}\right)^3.
    $$
By a union bound over all %
$ 
c_2\log (1/\eps)
$ 
samples $\bw$, the probability that \Cref{algo:gamma_query} rejects is at most $\eps/(20c_0)$ given that $c_1$ is sufficiently larger than $c_0$ and $c_2$.
So the probability of either rejecting or returning $0$ is at most ${\eps/(10c_0)}$.
This finishes the proof of the lemma.%
\end{proof}

Finally, we prove \Cref{thm: when f is a DL}:
\begin{proof}
    Let $f:\zo^n \to \zo$ be a decision list. If Case 1 holds, then by \Cref{lem: f is a DL case 1} the algorithm accepts $f$ with probability at least $0.9$. Otherwise, $f$ falls in Case 2. %
By \Cref{lem: U and R are good} and \Cref{lem: H is good},
  with probability at least $0.9$, the algorithm either accepts
  $f$ in line $1$ (in which case we are done) or reaches line 4 with $U,R$ and $u$
  that satisfy \Cref{assumption: U R are good}. 

Under \Cref{assumption: U R are good}, by \Cref{lem: f is DL 6-7} $f$ gets rejected in lines {3--6}
with probability at most $0.1$,
and gets rejected in line {7}
by \hyperlink{Algorithm2}{\sc Conj-Test} (including the case when it is rejected
  by \Cref{algo:gamma_query}) with probability at most $0.1$.
    So altogether,   $f$ is rejected with probability at most $0.3$. %
\end{proof}

\subsection{Soundness: 
 Correctness of \Cref{algo: Decision list tester} when $f$ is far from decision lists}\label{sec:soundness}

Let $f:\{0,1\}^n\rightarrow \{0,1\}$ be a function that is $\eps$-far from
  decision lists in relative distance.
Given~that conjunctions are special cases of 
  decision lists,
  by \Cref{thm:conjunction}, %
    \hyperlink{Algorithm2}{\sc Conj-Test} 
    on line 1 rejects with probability at least $0.9$, %
    in which case the algorithm continues 
    and reaches line 3.
Let $\bU=U$, $\bR=R$ and $\bu=u$ when it reaches line 3. After $U,R$ and $u$
  are fixed, the function  
  $\bGamma=\Gamma$ is also fixed.

If $f$ gets rejected on line 3 or 4 we are done;
  the next lemma gives a necessary condition
  for $f$ to get through lines 3--4 without 
  being rejected {with high probability}:

\begin{lemma}\label{thm: large premimage fu}
Assume that $f$ reaches line 3 with $U,R$ and $u$. 
If either 
\begin{equation}\label{eq:hehe1}\left|f{\upharpoonleft_{u}}^{-1}(1)\right|<\frac{2^{|R|}}{8}
\end{equation}
or $\reldist(f{\upharpoonleft_{u}},\textsf{DL})\ge \eps/10$, %
 then $f$ gets rejected during lines 3--4 with probability at least $0.9$.
\end{lemma}
\begin{proof}
If \Cref{eq:hehe1} holds then the probability of $f(u\circ \bw)=1$ when $\bw\sim \{0,1\}^R$ is at most $1/8$.
Given the $1/4$ vs $1/8$ gap and that $c_2$ is sufficiently large,
    \Cref{algo: Decision list tester} rejects on line 3 with probability at least $0.9$.

Now assume that \Cref{eq:hehe1} does not hold but $\reldist(f{\upharpoonleft_{u}},\textsf{DL})\ge \eps/10$.
We show that these two conditions together imply 
  $\dist(f{\upharpoonleft_{u}},\textsf{DL})\ge \eps/80$ and thus, by \Cref{thm:Bshouty20}
  the algorithm rejects on line 4 with probability
  at least $0.9$.
To lower bound $\dist(f{\upharpoonleft_{u}},\textsf{DL})$, we use
    $${\frac \eps {10}} \le \reldist(f{\upharpoonleft_{ u}},\textsf{DL})= %
     \dist(f{\upharpoonleft_{ u}},\textsf{DL})\cdot \frac{2^{|R|}}{\big|f{\upharpoonleft_{ u}}^{-1}(1)\big|}\leq 8\cdot \dist(f{\upharpoonleft_{ u}},\textsf{DL}).$$
This finishes the proof of the lemma.
\end{proof}

Now we look at lines 5--6. 
The next lemma gives a necessary condition
  for $f$ not getting rejected on line {6}
with high probability.
To this end, we introduce a new function
$g:\{0,1\}^n\rightarrow \{0,1\}$:
$$
g(x)=\Gamma(x_U) {\wedge} f(u\circ x_R).
$$  
By definition, for each $\alpha\in \{0,1\}^U$, $g{\upharpoonleft_{\alpha}}$ is all-$0$ if $\Gamma(\alpha)=0$, and $g{\upharpoonleft_{\alpha}}=f{\upharpoonleft_{ u}}$ if $\Gamma(\alpha)=1$.
\begin{lemma}\label{thm: f close to g (DL case)}
Assume that $f$ reaches line 4 with $U,R$ and $u$ and satisfies 
  $ |f{\upharpoonleft_{u}}^{-1}(1)|\ge 2^{|R|}/8.$
If we have $\reldist(f,g) \geq \eps/10$,
  then $f$ gets rejected on line 6 with probability at least $0.9$.
\end{lemma}
\begin{proof}
Assume that $\reldist(f,g)\ge \eps/10$. 
We partition $f^{-1}(1)\hspace{0.06cm}\triangle\hspace{0.06cm} g^{-1}(1)$ into
\begin{align*}
P&:=%
\big\{y\in \{0,1\}^n: \Gamma(y_U)=0\ \text{and}\ f(y)=1\big\}%
\quad\text{and}\\[1ex]
Q&:=\big\{y\in \{0,1\}^n: \Gamma(y_U)=1\ \text{and}\ 
f(y)\ne g(y)\big\}.
\end{align*}
So we either have $|P|\ge (\eps/20)\cdot |f^{-1}(1)|$ or $|Q|\ge (\eps/20)\cdot |f^{-1}(1)|.$

First consider the case when $P$ is large.
Each round of line 6 rejects with probability at least
$$
 \frac{\eps}{20} \cdot \frac{1}{16},
$$
where $\eps/20$ is the probability of $\bz$ landing in $P$.
To see where $1/16$ comes from, note that whenever $z\in P$, we have $\Gamma( z_U) =0$ and thus $|f{\upharpoonleft_{z_U}}^{-1}(1)|\le 2^{|R|}/16$
  while $|f{\upharpoonleft_{u}}^{-1}(1)|\ge 2^{|R|}/8$ by assumption.
The volume difference between $f{\upharpoonleft_{ u}}$
  and $f{\upharpoonleft_{ z_U}}$ shows that $f(u\circ \bw)\ne f( z_U\circ \bw)$ with probability at least $1/16$.
By setting $c_2$ sufficiently large, line 6 rejects with probability at least $0.9$.

Next for the case when $Q$ is large,
the probability that each round of line {6}
rejects is at least
$$
\sum_{y\in Q}
 \frac{ |f{\upharpoonleft_{ y_U}}^{-1}(1)|}{|f^{-1}(1)|}\cdot \frac{1}{2^{|R|}}.
$$
{Here for each $y \in Q$, the term above is exactly the probability that 
  $\bz\sim \SAMP(f)$ and $\bw\sim \{0,1\}^R$ of line 6 together satisfy
  $\bz_U\circ \bw=y$, in which case line 6 rejects.}
The numerator is at least $2^{|R|}/16$ given that  $\Gamma(z_U)=1$.
Plugging it in, each round of line {6} 
rejects with probability at least 
$$
|Q| \cdot \frac{2^{|R|}/16}{|f^{-1}(1)|}\cdot \frac{1}{2^{|R|}}\ge 
\frac{\eps}{20}\cdot |f^{-1}(1)|\cdot \frac{2^{|R|}/16}{|f^{-1}(1)|}\cdot \frac{1}{2^{|R|}}\ge \frac{\eps}{320}.
$$
We again finish the proof by setting the constant
  $c_2$ to be sufficiently large.
\end{proof}

Finally, we look at Step 4 of the tester. 
Given the lemmas established so far, we assume that $f$ reaches line 7 with $U,R$ and $u$ (as well as $g$ and $\Gamma$) that satisfy the following assumption:

\begin{assumption}\label{assumption: fu dense + close to g}
$f$ reaches line 7 with $U,R,u$ and $g$ that satisfy the following conditions:
    \begin{itemize}
    \item $\reldist(f,\textsf{DL})\ge \eps$;
        \item $\big|f{\upharpoonleft_{ u}}^{-1}(1)\big| \geq 2^{| R|}\big/8 $;
        \item $\reldist(f{\upharpoonleft_{u}},\textsf{DL})\le \eps\big/10 $; and
        \item $\reldist(f, g)\leq \epsilon\big/10$.
    \end{itemize}
\end{assumption}

We need to show that under \Cref{assumption: fu dense + close to g}, 
\hyperlink{Algorithm2}{\sc Conj-Test} on line 7
  rejects with high probability. 
We proceed in two steps.
\begin{flushleft}\begin{enumerate}
\item First we show in \Cref{thm: f would be close to DL} that if $f$ satisfies \Cref{assumption: fu dense + close to g},
  then $\Gamma$ must be $(\eps/100)$-far from conjunctions in
  relative distance. 
\item Next we show in \Cref{thm: gamma is rejected whp} that \hyperlink{Algorithm2}{\sc Conj-Test} on line 7 rejects with high probability. Note that this would follow trivially if \hyperlink{Algorithm2}{\sc Conj-Test} were given access to sampling and membership oracles of $\Gamma$.
However, \hyperlink{Algorithm2}{\sc Conj-Test} is only given access simulated by \Cref{algo:gamma_query} and \Cref{algo:gamma_sampler},
  respectively. 
So we analyze the their performance guarantees in \Cref{thm: no query mistake when far}
  and \Cref{thm: samples are almost uniform}, respectively, and then apply \Cref{thm: reject weaker assumption} to prove \Cref{thm: gamma is rejected whp}.
\end{enumerate}\end{flushleft}

We start with the first step:
\begin{lemma}\label{thm: f would be close to DL}
Under 
  \Cref{assumption: fu dense + close to g}, $\Gamma$
  is $(\eps/100)$-far from conjunctions in relative distance.
\end{lemma}
\begin{proof}
Assume for a contradiction that $\reldist(\Gamma,C)\le \eps/100$
  for some conjunction $C$ over $\{0,1\}^U$.
We prove in the rest of the proof that $\reldist(f,\textsf{DL})< \eps$, a contradiction.

    The proof proceeds in two steps, where we define 
      two Boolean functions $\zeta$ and $\lambda$ over $\{0,1\}^n$.
    In particular, the last function $\lambda$ will be a decision list.
    We prove the following inequalities: 
$$  \reldist(g,\zeta)\le \eps/10 \quad\text{and}\quad\reldist(\zeta,\lambda)\le \eps/100 .$$
Combining with $\reldist(f,g)\le \eps/10 $ from
  \Cref{assumption: fu dense + close to g},  it follows from  
  the approximate triangle inequality for relative distance (\Cref{thm: approx triangle ineq}) %
  that  $\reldist(f,\lambda)<\eps$, a contradiction.

First, given $\reldist(f{\upharpoonleft_{u}},\textsf{DL})\le \eps/10 $ from \Cref{assumption: fu dense + close to g}, we let $L$ be a decision list over $\{0,1\}^R$ that satisfies $\reldist(f{\upharpoonleft_{u}},L)\le \eps/10 $.
Then our first function
$\zeta$ over $\{0,1\}^n$ is defined as %
    $$\zeta(z) :=\Gamma(z_U){\wedge} L(z_{ R}).$$    
Observe that for any $\alpha \in \{0,1\}^{U}$, if $\Gamma(\alpha)=0$ then $g{\upharpoonleft_\alpha}=\zeta{\upharpoonleft_\alpha}=0$. And if $\Gamma(\alpha)=1$, then $g{\upharpoonleft_\alpha}=f{\upharpoonleft_{ u}}$ while $\zeta{\upharpoonleft_\alpha}=L$, in which case, ${\reldist(g{\upharpoonleft_\alpha}, \zeta{\upharpoonleft_\alpha} )}=\reldist(f{\upharpoonleft_{ u}}, L) \leq \epsilon/10 $. We thus have that:
     \begin{align*}
        \reldist(g, \zeta)&=\frac{1}{|g^{-1}(1)|} \sum_{\alpha: \Gamma(\alpha)=1} \left| \big\{w \in \{0,1\}^R:g(\alpha\circ w) \neq \zeta(\alpha\circ w)\big\}\right| \\
        &= \frac{1}{|g^{-1}(1)|} \sum_{\alpha:\Gamma(\alpha)=1} \reldist(g{\upharpoonleft_\alpha}, \zeta{\upharpoonleft_\alpha}) \cdot \big|g{\upharpoonleft_\alpha^{-1}}(1)\big| \\
        &\leq \frac{\eps}{10 }\cdot \frac{1}{|g^{-1}(1)|} \sum_{\alpha:\Gamma(\alpha)=1} \big|g{\upharpoonleft_\alpha^{-1}}(1)\big|.
    \end{align*}
Since $|g^{-1}(1)|=\sum_{\alpha }|g{\upharpoonleft_\alpha^{-1}}(1)|$, we get that $\reldist(g, \zeta) \leq \epsilon/10$.

Next recall that $C$ is a conjunction over $\{0,1\}^U$ that satisfies
  $\reldist(\Gamma,C)\le \eps/100$.
The last function $\lambda:\{0,1\}^n\rightarrow \{0,1\}$ is obtained  as follows:
    $$\lambda(z):=C(z_U) {\wedge} L(z_R)$$
    (i.e.~we are swapping out $\Gamma(z_U)$ for $C(z_U)$ in the definition of $\zeta$).
    In particular, %
    $\lambda$ is a decision list as promised earlier in the proof.
    Below we bound $\reldist(\zeta,\lambda)$.%

    Observe that for any $\alpha \in \{0,1\}^{ U}$, if $\Gamma(\alpha)=C(\alpha)$, then $\lambda{\upharpoonleft_\alpha}=\zeta{\upharpoonleft_\alpha}$. Otherwise, we have that~one of  $\lambda{\upharpoonleft_\alpha}$ or $\zeta{\upharpoonleft_\alpha}$ is $L$ while the other is the all-$0$ function. 
    So we have 
    \begin{align*}
        \reldist(\zeta, \lambda)&=\frac{1}{\left|\zeta^{-1}(1)\right|} \cdot \left|\big\{\alpha\in \{0,1\}^U: \Gamma(\alpha)\ne C(\alpha)\big\}\right|\cdot \big|L^{-1}(1)\big|.
    \end{align*}
On the other hand, $ 
         |\zeta^{-1}(1) |
         = | \Gamma^{-1}(1) |\times  |L^{-1}(1) |.
    $  
    Thus, we have
    \begin{align*}
          \reldist(\zeta, \lambda)
          = \frac{ | \{\alpha\in \{0,1\}^U: \Gamma(\alpha)\ne C(\alpha) \} |}{|\Gamma^{-1}(1)|}= \reldist(\Gamma,C)\le \frac{\epsilon}{100}.
    \end{align*}

This finishes the proof of the lemma.
\end{proof}

Next we analyze performance guarantees of \Cref{algo:gamma_sampler}
  and \Cref{algo:gamma_query} to simulate sampling and membership
  oracles of $\Gamma$, respectively:

\begin{lemma}\label{thm: samples are almost uniform}
    Assume that $f$ respects \Cref{assumption: fu dense + close to g}. 
    For any $\alpha\in \{0,1\}^U$ with $\alpha \in \Gamma^{-1}(1)$, we have
$$
\Pr\big[\text{\Cref{algo:gamma_sampler} returns }\alpha\big] \geq \frac{1}{20|\Gamma^{-1}(1)|}.%
$$
\end{lemma}
\begin{proof}
Let $\alpha\in \{0,1\}^U$ with $\Gamma(\alpha)=1$.
By definition of $\Gamma$, we have
    \begin{equation}
        \Prx_{\bz \sim f^{-1}(1)}\big[\bz_U=\alpha\big] %
        \ge \frac{2^{|R|}/16}{|f^{-1}(1)|}.\label{eq:a}
    \end{equation}
    On the other hand, 
      we know that
      $$
      \frac{\eps}{10}\ge \reldist(f,g)\ge 
      \frac{|f^{-1}(1)|-|g^{-1}(1)|}{|f^{-1}(1)|}
      $$
      and $|g^{-1}(1)|\le 
  |\Gamma^{-1}(1)|\cdot 2^{|R|}$.
Combining these two we have
    \begin{equation}
    |f^{-1}(1)|\le \frac{|g^{-1}(1)|}{1-(\eps/10)}\le  \frac{|\Gamma^{-1}(1)|\cdot 2^{|R|}}{1-(\eps/10)} \le 1.2\cdot |\Gamma^{-1}(1)|\cdot 2^{|R|}.\label{eq:b}
    \end{equation}
The lemma follows by combining \Cref{eq:a} and \Cref{eq:b}.
\end{proof}

\begin{lemma}\label{thm: no query mistake when far}
    Assume $f$ respects \Cref{assumption: fu dense + close to g}. For any $\alpha\in \{0,1\}^U$, 
    \Cref{algo:gamma_query} run on $\alpha$ satisfies
    $$
    \Pr\big[\text{\Cref{algo:gamma_query} does not reject on line {2}
     and  returns $1-\Gamma(\alpha)$}\big]
    \le \frac{\epsilon}{40c_0}.$$
\end{lemma}
\begin{proof}
If $\Gamma(\alpha)=0$, then by the definition of $\Gamma$ we have $|f{\upharpoonleft_\alpha}^{-1}(1)|\le 2^{|R|}/16.$ %
So {by the second bullet of \Cref{assumption: fu dense + close to g}, we have}
    $$\Prx_{\bw \sim \{0,1\}^{\bR}}\big[f(\alpha \circ \bw) \neq f(u \circ \bw)\big] \geq \frac{1}{16}.$$ 
The probability of \Cref{algo:gamma_query} 
  returning $1$ without rejecting $f$  is at most
$$
\left(\frac{15}{16}\right)^{c_2\log (1//\eps)}\le \frac{\eps}{40c_0}
$$
by making $c_2$ sufficiently large and also sufficiently larger than $c_0$.

If $\Gamma(\alpha)=1$, 
  the algorithm returns $0$ if all samples  $\bw$  had $f(\alpha\circ \bw)=0$.
By the definition of $\Gamma$, this happens with probability at 
  most $(15/16)^{c_2\log (1/\eps)}$.
\end{proof}

We are now ready to show that $f$
  is rejected on line 7 with high probability:
\begin{lemma}\label{thm: gamma is rejected whp}
    Assume $f$ respects \Cref{assumption: fu dense + close to g}. Then  \hyperlink{Algorithm2}{\sc Conj-Test}
    rejects on line $7$ of \Cref{algo: Decision list tester} with probability at least $0.85$.
\end{lemma}
\begin{proof}
By \Cref{thm: f would be close to DL}, $\Gamma$ is $(\eps/100)$-far from conjunctions in relative distance.
By \Cref{thm: samples are almost uniform}, every $\alpha \in \Gamma^{-1}(1)$ is returned by \Cref{algo:gamma_sampler}  with probability at least $1/(20|\Gamma^{-1}(1)|)$.
If all the $c_0/\eps$ many membership queries made by \hyperlink{Algorithm2}{\sc Conj-Test} are answered correctly, then
  by \Cref{thm: reject weaker assumption} %
  \hyperlink{Algorithm2}{\sc Conj-Test} rejects $\Gamma$ 
  with probability at least $0.9$.

Now consider an execution of \hyperlink{Algorithm2}{\sc Conj-Test}
  under the sampling oracle simulated by \Cref{algo:gamma_sampler} together with the correct membership query oracle,
  and the $c_0/\eps$ membership queries it makes.
If all these $c_0/\eps$ membership queries are sent to
  \Cref{algo:gamma_query}, then
by \Cref{thm: no query mistake when far} and a union bound, with
  probability at least $39/40$, either $f$ is rejected on line {2}
of \Cref{algo:gamma_query} or all $c_0/\eps$ membership queries are 
  answered correctly.
This shows that with probability at least $39/40$, either
  $f$ is rejected or \hyperlink{Algorithm2}{\sc Conj-Test} returns the same answer as it would with the correct membership oracle.
As a result, line 7 of \Cref{algo: Decision list tester} rejects with probability at least $0.9 \cdot (39/40)\ge 0.85$.  %
\end{proof}

Finally we can prove our result establishing soundness:

\begin{theorem} \label{thm:soundness}
    If $f$ is $\epsilon$-far from any decision list in relative distance, then \Cref{algo: Decision list tester} rejects it with probability at least $0.7$. 
\end{theorem}
\begin{proof}
Given that $f$ is $\eps$-far from decision lists, line 1 accepts
  with probability at most $0.1$ by \Cref{thm:conjunction}.
When this does not happen, the algorithm reaches line {3}
and we consider two cases.
  
If $f$ does not satisfy \Cref{assumption: fu dense + close to g},
  then by \Cref{thm: large premimage fu} and \Cref{thm: f close to g (DL case)} it is rejected during lines 3--6 with probability at least $0.9$.
If $f$ satisfies \Cref{assumption: fu dense + close to g}, then it follows from \Cref{thm: gamma is rejected whp} that it is rejected on line 7 with probability at least $0.85.$

As a result, the algorithm rejects $f$ with probability
  at least $0.9 \cdot 0.85\ge 0.7.$
\end{proof}

\begin{flushleft}
\bibliographystyle{alpha}
\bibliography{allrefs}
\end{flushleft}

\appendix

\def\Dyes{\calD_{\textrm{yes}}}
\def\Dno{\calD_{\textrm{no}}}
\def\bff{\boldsymbol{f}}
\section{Lower Bounds}\label{appendix:lower}

We follow ideas of \cite{BshoutyGoldreich022}
  to give an $\Omega(1/\eps)$ lower bound for any two-sided adaptive relative-error testing algorithm for conjunctions and decision lists.

To this end we define the following two distributions $\Dyes$ and $\Dno$:
\begin{flushleft}\begin{enumerate}
    \item $\Dyes$ is supported on the constant-$1$ function only; and \item $\bff\sim\Dno$ is drawn by setting independently every entry $\bff(x)$ of $\bff$ to be $0$ with probability $3\eps$\\ and $1$ with probability $1-3\eps.$
\end{enumerate}\end{flushleft}
The constant-$1$ function is trivially both a conjunction and 
  a decision list.
On the other hand, the following lemma shows that, as long as $\eps$ satisfies $C\cdot {\log n}/2^n\le \eps\le 0.1$ for some sufficiently large constant $C$,  
  $\bff\sim \Dno$ is far from decision lists (and thus, conjunctions) with high probability. 

\begin{lemma}
With probability at least $1-o_n(1)$,
  $\bff\sim \Dno$ satisfies $\reldist(\bff,\textsf{DL})\ge \eps.$
\end{lemma}
\begin{proof}
By setting $C$ suitably large and a Chernoff bound,
  with probability  $1-o_n(1)$, we have
\begin{enumerate}
    \item $|\bff^{-1}(0)|$ 
  lies in $(1\pm 0.1)(3\eps 2^n)$; and
  \item For any $i\in [n]$ and $b\in \{0,1\}$, we have $\big|\{x: x_i=b\ \text{and}\ \bff(x)=0\}\big|\ge 0.9(3\eps 2^{n-1})$.
\end{enumerate}
We show that $\reldist(\bff,\textsf{DL})\ge \eps$ whenever these conditions hold. Consider any decision list $g$:
\begin{enumerate}
\item If $g$ starts with $v_{\textsf{default}}=1$ then $g$ is the constant-$1$ function and we have $\dist(\bff,g)\ge 2.7\eps $;
\item If $g$ starts with $v_{\textsf{default}}=0$ or $(x_{i_1},b_1,0)$, then $|g^{-1}(1)|\le 2^{n-1}$ and $\dist(\bff,g)\ge 0.5-3.3\eps$;
\item If $g$ starts with $(x_{i_1},b_1,1)$,
  then by the second condition, we have 
  $\dist(\bff,g)\ge 1.35\eps$.
\end{enumerate}
The lemma then follows from $\reldist(\bff,g)\ge \dist(\bff,g)$ 
  and $0.5-3.3\eps\ge \eps$.
\end{proof}

Let $q=c/\eps$ for some constant $c$. 
Assume for a contradiction that $\calA$ is a randomized, two-sided, adaptive relative-error testing algorithm
for conjunctions\hspace{0.1cm}/\hspace{0.1cm}decision lists, which receives $q$ samples from $\SAMP(f)$ and makes $q$ queries on $\MQ(f)$. Then by the correctness of the algorithm, we have
\begin{align*}
\Pr\big[\text{$\calA$ accepts the constant-$1$ function}\big]\ge 0.9\quad\text{and}\quad
\Prx_{\bff\sim \Dno} \big[\text{$\calA$ accepts
  $\bff$}\big]\le 0.1
\end{align*}
and thus, we have 
$$
\Pr\big[\text{$\calA$ accepts the constant-$1$ function}\big]-
\Prx_{\bff\sim \Dno} \big[\text{$\calA$ accepts
  $\bff$}\big]\ge 0.8.
$$
Given that $\calA$ is a distribution
  over deterministic algorithms, there is
  a deterministic algorithm $\calB$ (i.e., what it does after receiving the samples is deterministic) with the 
  same complexity such that 
$$
\Pr\big[\text{$\calB$ accepts the constant-$1$ function}\big]-
\Prx_{\bff\sim \Dno} \big[\text{$\calB$ accepts
  $\bff$}\big]\ge 0.8.
$$
We show below that when $c$ is sufficiently small, 
  the LHS is at most $0.1$, a contradiction.

To this end, consider an execution of $\calB$ on
  the constant-$1$ function.
First $\calB$ draws $q$ samples from $\{0,1\}^n$,
  which we denote by $\bX=(\bx^1,\ldots,\bx^q)$ and is 
  uniform over $(\{0,1\}^n)^q$.
After receiving~$\bX$, the deterministic algorithm $\calB$ makes a sequence of $q$ queries. While $\calB$ is adaptive, given that it is run on the constant-$1$ function, all queries return $1$ and thus, the $q$ queries made by $\calB$ is determined by $\bX$ and we denote them by $\bY=(\by^1,\ldots,\by^q)$.
We refer to a pair $(X,Y)$ as a transcript of $\calB$
  on the constant-$1$ function if $Y$ consists of 
  queries that $\calB$ makes after receiving $X$ as samples.
There are $2^{nq}$ transcripts $(X,Y)$, one for each
  $X\in (\{0,1\}^n)^q$, and $(\bX,\bY)$ is uniformly distributed among all the $2^{nq}$ transcripts.

Let $\calT$ denote the set of all transcripts that lead $\calB$ to accept the constant-$1$ function. 
Then 
\begin{equation}\label{eq:hehe99}
\Pr\big[\text{$\calB$ accepts the constant-$1$ function}\big]=|\calT|\big/2^{nq}.
\end{equation}

On the other hand, fix any $(X,Y)$ in $\calT$,
  where $X=(x^1,\ldots,x^q)$ and $Y=(y^1,\ldots,y^q)$, and
we consider the probabiblity of the following event:
\begin{quote}
$\calE_{X,Y}$: Draw $\bff\sim \Dno$; when running 
  $\calB$ on $\bff$, $\calB$ receives $X$ as the sequence of $q$ samples and the sequence of $q$ queries it makes
  are $Y=(y^1,\ldots,y^q)$ and all of them return $1$.
\end{quote}
We note that when $\calE_{X,Y}$ occurs,   
  $\calB$ must accept $\bff$.
As a result we have
$$
\Prx_{\bff\sim \Dno} \big[\text{$\calB$ accepts
  $\bff$}\big]\ge \sum_{(X,Y)\in \calT}
\Prx_{\bff\sim \Dno} \big[\calE_{X,Y}\big].
$$
Finally, for any $(X,Y)\in \calT$,
  by setting $c$ sufficiently small, the probability of $\calE_{X,Y}$ is at least %
$$
(1-3\eps)^{2q}\cdot \left(\frac{1}{2^n}\right)^q=(1-3\eps)^{2c/\eps}\cdot \frac{1}{2^{nq}}\ge 0.9\cdot \frac{1}{2^{nq}},
$$
where the $(1-3\eps)^{2q}$ on the LHS is the probability that every point $z\in X\cup Y$ has $\bff(z)=1$; when this happens, the probability of getting $X$ as the samples is
$$
\left(\frac{1}{|\bff^{-1}(1)|}\right)^q\ge \left(\frac{1}{2^n}\right)^q,
$$
and when this also happens, $Y$ must be the queries that $\calB$ makes and all queries return $1$.

Putting all pieces together, we have 
$$
\Pr\big[\text{$\calB$ accepts the constant-$1$ function}\big]-
\Prx_{\bff\sim \Dno} \big[\text{$\calB$ accepts
  $\bff$}\big]\le \frac{|\calT|}{2^{nq}}-0.9\cdot \frac{|\calT|}{2^{nq}}=0.1\cdot \frac{|\calT|}{2^{nq}}
$$
and using $|\calT|\le 2^{nq}$, the LHS is at most
  $0.1$, a contradiction.

\end{document}